\pdfoutput=1
\documentclass[11pt]{article}
\usepackage[final]{acl}
\usepackage{times}
\usepackage{latexsym}
\usepackage{amsfonts}
\usepackage{color}
\usepackage{soul}
\usepackage[T1]{fontenc}
\usepackage[utf8]{inputenc}
\usepackage{microtype}
\usepackage{algorithm}
\usepackage{enumitem}
\usepackage{algpseudocode}
\usepackage{amsmath}
\usepackage{amsthm}
\usepackage{array}
\usepackage{multirow}
\usepackage{amssymb}
\usepackage{booktabs}
\usepackage{listings}
\usepackage{tcolorbox}
\tcbuselibrary{listings, skins, breakable}
\newtcolorbox{promptlisting}[1][]{
  enhanced,
  breakable,
  sharp corners,          
  boxrule=0.6pt,         
  top=2pt, bottom=2pt, left=2pt, right=2pt,
  before skip=8pt,
  after skip=8pt,
  colframe=gray!75!black,    
  colback=gray!5,            
  colbacktitle=gray!75!black,
  coltitle=white,            
  titlerule=0pt,            
  fonttitle=\bfseries\small,
  title={#1},
  listing only,
  listing options={
    basicstyle=\fontfamily{pcr}\selectfont\footnotesize,
    breaklines=true,
    breakatwhitespace=true,
    columns=fullflexible,
    aboveskip=0pt,
    belowskip=0pt
  },
  extras broken={
    frame empty,
    overlay={
        \draw[gray!75!black, line width=0.6pt] (frame.north west) rectangle (frame.south east);
    }
  }
}
\tcbuselibrary{skins, breakable}
\newtcolorbox{fullcasebox}[1][]{
  enhanced,
  colframe=black,          
  colback=white,           
  colbacktitle=gray!10,    
  coltitle=black,          
  fonttitle=\bfseries\small,
  title={#1},
  titlerule=0.4pt,         
  titlerule style={gray!50},
  boxrule=0.6pt,           
  sharp corners,
  top=2pt, bottom=2pt, left=2pt, right=2pt,
  fontupper=\footnotesize,
  parbox=false,
  halign=flush left,
  breakable,
  before skip=8pt,
  after skip=8pt,
  extras broken={
        frame empty, 
        overlay={   
            \draw[black, line width=0.6pt] (frame.north west) rectangle (frame.south east);
        }
    },
}
\newcolumntype{Y}{>{\centering\arraybackslash}p{1.4cm}}
\usepackage{subcaption}
\usepackage{graphicx}
\theoremstyle{definition}
\newtheorem{definition}{Definition}
\newtheorem{assumption}{Assumption}
\theoremstyle{plain}
\newtheorem{proposition}{Proposition}
\newtheorem{corollary}{Corollary}
\makeatother
\usepackage{enumitem}
\setlist{nosep, leftmargin=*}

\title{Look Twice before You Leap: A Rational Framework for \\ Localized Adversarial Text Anonymization}
\author{
  Donghang Duan$^1$,\ Xu Zheng$^1$,\ Yuefeng He$^1$, \\ 
  \textbf{Chong Mu}$^1$,\ \textbf{Leyi Cai}$^2$,\ \textbf{Lizong Zhang}$^1$ \\
  $^1$School of Computer Science and Engineering, \\
  University of Electronic Science and Technology of China \\
  $^2$School of Information, Renmin University of China
}

\begin{document}
\maketitle
\begin{abstract}
Current LLM-based frameworks for text anonymization usually rely on remote API services from powerful LLMs, 
which creates an inherent privacy paradox: users must disclose the raw data to untrusted third parties for guaranteed privacy preservation. 
Moreover, directly migrating current solutions to local small-scale models (LSMs) offers a suboptimal solution with severe utility collapse.
Our work argues that this failure stems not merely from the capability deficits of LSMs, but significantly from the inherent irrationality of the greedy adversarial strategies employed by current state-of-the-art (SOTA) methods. 
To address this drawback, we propose Rational Localized Adversarial Anonymization (RLAA), a fully localized and training-free framework featuring an Attacker-Arbitrator-Anonymizer architecture. 
We model the anonymization process as a trade-off between Marginal Privacy Gain (MPG) and Marginal Utility Cost (MUC), demonstrating that greedy strategies tend to drift into an irrational state. 
Instead, RLAA introduces an arbitrator that acts as a rationality gatekeeper, 
validating the attacker's inference to filter out ghost leaks. 
This mechanism promotes a rational early-stopping criterion, and structurally prevents utility collapse.
Extensive experiments on different benchmarks demonstrate that RLAA achieves a superior privacy-utility trade-off compared to strong baselines.
\end{abstract}

\section{Introduction}
Large language models (LLMs) are increasingly deployed to process real-world text 
containing sensitive Personal Identifiable Information (PII), such as medical records, 
legal documents and online self-disclosures \cite{bommasani2021opportunities, weidinger2021ethical, deusser2025survey}. 
To comply with regulatory requirements like GDPR and CCPA \cite{albrecht2016gdpr, bonta2022california}, 
effective text anonymization has become a prerequisite for the responsible use of such data.
However, this situation presents a tricky trade-off between semantic utility and privacy,
where over-anonymization destroys the semantic value while under-anonymization invites re-identification risks,
which is particularly acute in sensitive domains like smart healthcare and legal technology \cite{im2024exploring, morris2025diri, liu2025legal}.  
The rapid evolution of LLMs has exacerbated this situation because of their dual roles as context-aware anonymizers and superior attackers \cite{wang2025unique}. 
As a result, traditional NER-based anonymization methods are increasingly inadequate \cite{dernoncourt2017identification}, 
motivating a shift toward semantic anonymization that conceals latent cues of identity beyond surface-level entities \cite{loiseau2025tau}.

To address this situation, advanced anonymization frameworks based on LLMs have emerged.
As the current state-of-the-art paradigm, Feedback-guided Adversarial Anonymization (FgAA) \cite{staab2024large} employs a dynamic adversarial game: 
an attacker model attempts to re-identify personal information from the text generated by an anonymizer model, 
and the anonymizer uses the inference feedback to instruct the anonymization process for refinement. 
Despite its strong empirical performance, FgAA heavily relies on the capabilities of powerful or closed-source LLMs like GPT-4 \cite{achiam2023gpt}, 
which are typically accessible only via third-party APIs.
This dependency gives rise to a fundamental privacy paradox: 
to anonymize sensitive data, users must first disclose the raw content to external and uncontrollable service providers, 
which renders it unacceptable in high-sensitivity scenarios \cite{feretzakis2024trustworthy}.
Although subsequent works aim to mitigate this risk through localized deployment, they face distinct limitations. 
IncogniText \cite{frikha2024incognitext} compromises semantic utility by injecting ungrounded hallucinations. 
SEAL \cite{kim2025self} necessitates complex training pipelines that rely on high-quality synthetic datasets which remain scarce \cite{yukhymenko2024synthetic}.
Besides, our empirical analysis in Section \ref{sec:performance_analysis} reveals that it still suffers from severe utility collapse.

\begin{figure}[t]
    \centering
    \includegraphics[width=0.4\textwidth]{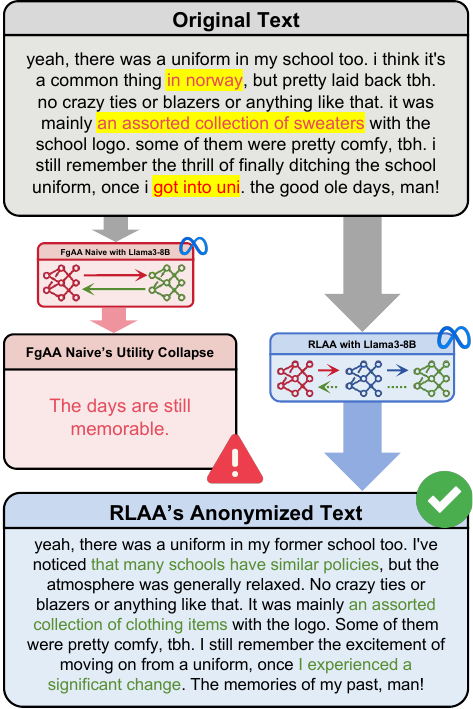}
    \vspace{0.2cm}
    \caption{Utility Collapse of FgAA’s Naive Migration.}
    \label{fig:utility_failure}
\end{figure}

A natural and simple solution to the privacy paradox is to migrate adversarial anonymization frameworks to fully localized environments using local small-scale models (LSMs).
However, our empirical investigation reveals that such a naive migration leads to severe utility collapse.
As illustrated in Figure \ref{fig:utility_failure}, the naive migration on Llama3-8B causes destructive over-editing. 
It indiscriminately strips away non-sensitive context and stylistic nuances, reducing the expressive original narrative to a generic and vacuous summary.
Besides, it is observed that this collapse persists even when implementing instruction refinement (\textit{e.g.}, outputting "unknown" for inferences with low confidence).
This observation demonstrates that heuristic constraints alone cannot prevent the utility collapse in greedy adversarial strategies.

We argue that this failure cannot be solely attributed to the limited capabilities of LSMs,
but instead arises from how greedy adversarial strategies behave under imperfect inferences.
From an economic perspective, the anonymization can be viewed as a sequence of decisions balancing \textbf{Marginal Privacy Gain (MPG)} against \textbf{Marginal Utility Cost (MUC)}.
By quantifying this ratio as the \textbf{Marginal Rate of Substitution ($\text{MRS}= \frac{\text{MUC}}{\text{MPG}}$)} \cite{mas1995microeconomic}, the primary driver of the utility collapse becomes clear through this economic lens: 
Current greedy strategies tend to drift into an economically inefficient state due to \textbf{LSMs' hallucinations} and \textbf{the law of diminishing returns}, 
where the model behaves irrationally by exhibiting exaggerated defense \cite{rottger2024xstest} against \textbf{ghost leaks} that are either hallucinated or negligible (detailed in Appendix \ref{app:theoretical_analysis}).
In this case, MUC remains positive while MPG approaches zero, pushing MRS to inefficient levels.

To overcome both the data dependency of distillation-based approaches and the utility collapse of naive migration, we propose \textbf{RLAA (Rational Localized Adversarial Anonymization)\footnote{Our code and datasets are available at \url{https://github.com/SowingG2333/RLAA}.}}, a training-free framework designed for fully localized deployment. 
The core idea of RLAA is to constrain the anonymization process to the economically rational condition defined in Section \ref{sec:problem_formulation}, 
thereby compensating for the capability deficits of LSMs without relying on explicit numerical optimization or parameter fine-tuning.
RLAA introduces a novel \textbf{Attacker–Arbitrator–Anonymizer (A–A–A)} architecture, in which the arbitrator acts as a rationality gatekeeper between attack feedback and anonymization actions. 
Moving beyond the paradigm of blindly leaping into modifications, RLAA compels the arbitrator to validate the attacker's inferences, which structurally prevents utility collapse by rejecting destructive edits driven by ghost leaks, and maintains robust privacy protection by focusing on valid leaks.
Crucially, this design leverages the \textbf{cognitive asymmetry} that verification is less complex and more reliable than generation for LSMs, which may hallucinate during open-ended inference but retain the ability to recognize errors in structured discrimination tasks \cite{cobbe2021training, wang2022self, madaan2023self, guan2024language}.

We conducted a comprehensive evaluation including baseline comparisons, ablation studies and generalization stress tests across several mainstream LLMs (Llama3-8B \cite{llama3modelcard}, Qwen2.5-7B \cite{qwen2.5} and DeepSeek-V3.2-Exp \cite{deepseekai2024deepseekv32}) on the PersonalReddit and reddit-self-disclosure datasets. 
Additionally, we validated the framework's economic rationality through a quantitative analysis based on the MRS metric. 
The results empirically demonstrate that RLAA structurally prevents utility collapse and achieves a superior privacy-utility trade-off compared to strong baselines.
In summary, the main contributions of this study are as follows:
\begin{itemize}[leftmargin=*, nosep]
    \item We argue that migrating adversarial anonymization to local environments is crucial for eliminating the privacy paradox, while a naive migration results in utility collapse, identified as a symptom of the economic irrationality inherent in greedy strategies, rather than merely capability deficits.
    \item RLAA is proposed as a localized and training-free framework featuring an Attacker-Arbitrator-Anonymizer (A-A-A) architecture. 
    By introducing an arbitrator, it structurally promotes rational decision-making to prevent utility collapse, providing a structured mechanism to preserve utility while reducing privacy risks without fine-tuning.
    \item Extensive experiments demonstrate that RLAA achieves a superior privacy-utility trade-off compared to multiple strong baselines and even Pareto dominance on the reddit-self-disclosure dataset, while also proving its mechanism generalization across different base models.
\end{itemize}

\section{Related Work}
\subsection{Inference Attacks through LLMs}
With increased reasoning capabilities, LLMs have evolved into potent privacy attackers capable of automatically inferring personal attributes from unstructured text \cite{wang2025unique}. 
\citet{staab24beyond} demonstrated that models can accurately deduce fine-grained PII from casual online comments. 
Extending to high-stakes domains, \citet{nyffenegger2024anonymity} revealed that LLMs could re-identify individuals in court decisions.
Recent audits by \citet{panda2025daiq} further confirmed that this inference capability is consistent across various model architectures.
These studies show that LLMs significantly lower the threshold for privacy breaches, enabling re-identification attacks to occur on an unprecedented scale \cite{staab24beyond}.

\subsection{Text Anonymization}
Recent advancements in text anonymization have marked a fundamental evolution from early statistical obfuscation techniques toward sophisticated semantic rewriting mechanisms. 
Generally, these studies can be divided into two categories:

\noindent\textbf{Traditional Approaches.}
In early efforts, Differential Privacy (DP) methods \cite{abadi2016deep, wu2023privacy, igamberdiev2023dp} and representation learning \cite{coavoux2018privacy, friedrich2019adversarial} served as the standard. 
However, these methods often compromise textual readability or require training specific encoders, which limits their applicability to open-ended generation tasks.
To address utility preservation, \citet{shetty2018a4nt} introduced A4NT to obfuscate author attributes, while \citet{mosallanezhad2019deep} utilized deep reinforcement learning to optimize the privacy-utility trade-off. 
These works laid the foundational mechanism for modern anonymizers.

\noindent\textbf{LLM-driven Adversarial Frameworks.}
With the rise of generative AI, \textbf{Feedback-guided Adversarial 
Anonymization (FgAA)} \cite{staab2024large} established the prevailing 
paradigm for text anonymization. It systematizes the LLM's 
duality, leveraging an attacker model to iteratively critique and refine 
the anonymizer's output. 
Following this paradigm, \textbf{SEAL} \cite{kim2025self} attempts to distill these capabilities into smaller models through SFT and DPO from multi-round teacher trajectories. 
However, this mechanism remains constrained by synthetic data scarcity \cite{yukhymenko2024synthetic} and potential generalization issues. 
\textbf{IncogniText} 
\cite{frikha2024incognitext} employs a misleading strategy, injecting false attributes to randomize attacker inferences rather than simply removing PII.
Beyond general-purpose anonymization, \textbf{RUPTA} \cite{yang2025robust} 
addresses a complementary problem: optimizing anonymized text for specific 
downstream tasks (\textit{e.g.}, occupation classification) using specialized task evaluators. 
While orthogonal to our focus on general semantic preservation, RUPTA shares the common reliance on API-based LLMs.
Crucially, a common limitation across these advanced frameworks is their reliance on strong LLMs via APIs or extensive training pipelines.
Moreover, methods following the greedy adversarial paradigm lack rational 
decision-making mechanisms, making them prone to utility collapse when 
deployed on smaller local models.

\section{Methodology}
\subsection{Threat Model}
\label{sec:threat_model}
RLAA is designed to defend against two distinct adversaries in the text anonymization task:

\noindent\textbf{Semi-honest Service Provider ($\mathcal{A}_{serv}$).}
This adversary represents a third-party entity that hosts remote LLM services. 
It causes the privacy paradox: to utilize superior anonymizers, users must first transmit raw sensitive text $x^{ori}$ to the provider. 
This transmission exposes users to risks such as unauthorized data retention for model training or commercial profiling \cite{smith2023identifying}.
Distinct from probable inference attacks, we characterize this threat as a deterministic data exposure, where $\mathcal{A}_{serv}$ gains full access to $x^{ori}$ upon submission.

\noindent\textbf{Re-identification Adversary ($\mathcal{A}_{re-id}$).}
This adversary represents any entity with access to the anonymized output $x^{ano}$.
We assume $\mathcal{A}_{re-id}$ employs a powerful LLM $\mathcal{M}_{atk}$ (\textit{e.g.}, DeepSeek-V3.2-Exp in our evaluation) to infer PII. 
The adversary's goal is to maximize the accuracy of the inferred attribute-value tuples $\{(a_j, v_j)\}$, where $a_j$ denotes a specific private attribute category (\textit{e.g.}, location) and $v_j$ represents its corresponding inferred value (\textit{e.g.}, "Paris, France"):
\begin{equation}
    \small
    \{(a_j, v_j)\} \leftarrow \mathcal{M}_{atk}(x^{ano})
\end{equation}

\subsection{Problem Formulation}
\label{sec:problem_formulation}
The text anonymization task can be formalized as a constrained transformation problem. 
Let $\mathcal{X}$ be the space of all text sequences. 
Given an original text $x^{(0)} \in \mathcal{X}$, let $\mathcal{A}_{priv} = \{(a_k, v_k^*)\}_{k=1}^K$ be the set of "private attribute-true value" pairs contained in $x^{(0)}$, where $a_k$ represents the attribute category (\textit{e.g.}, age, location) and $v_k^*$ denotes its true value. 
Let $\mathcal{M}_{atk}$ be a re-identification adversary. 
We define the privacy compromise metric $P_{asr}(x)$ as the adversary's success rate in inferring $\mathcal{A}_{priv}$ from a text $x$:
\begin{equation}
\small
\resizebox{0.9\linewidth}{!}{$
    P_{asr}(x) = \frac{1}{K} \sum_{k=1}^{K} \mathbb{I}\left( \hat{v}_k \approx v^*_k \mid (a_k, \hat{v}_k) \in \mathcal{M}_{atk}(x) \right)
$}
\end{equation}
where $\mathbb{I}(\cdot)$ is the indicator function checking if the inferred value $\hat{v}_k$ semantically aligns with the ground truth $v_k^*$. 
This fuzzy matching operator ($\approx$) accounts for the inherent linguistic variability in LLM-generated inferences, ensuring the evaluation captures semantic equivalence rather than strict string identity.

Simultaneously, let $U(x^{(0)}, x): \mathcal{X} \times \mathcal{X} \to \mathbb{R}$ be a utility function measuring semantic preservation, where a higher value indicates better utility.
The ideal objective is to find an optimum $x^*$ such that privacy is protected while utility is maximized:
\begin{equation}
    \max \quad U(x^{(0)}, x) \quad \text{s.t.} \quad P_{asr}(x) \le \delta
    \label{eq:problem_formulation}
\end{equation}
where $\delta$ denotes the upper bound on acceptable privacy risk.

Migrating the framework to a local environment naturally eliminates the service provider threat $\mathcal{A}_{serv}$.
However, achieving the objective in Eq. \ref{eq:problem_formulation} is non-trivial. 
Existing adversarial frameworks typically approximate the optimum $x^*$ via an iterative sequence $x^{(0)} \to x^{(1)} \dots \to x^{(T)}$, where each step attempts to reduce privacy risk.
Our empirical investigation reveals a critical challenge: directly applying greedy adversarial iterations with LSMs in this process often leads to utility collapse.
We hypothesize that this failure arises because naive greedy strategies lack a mechanism to evaluate the cost-effectiveness of privacy gains.
To formally address this, we reframe this iterative process over steps $t = 1, 2, \dots T$ through an economic lens, viewing each anonymization operation as a transaction.
The following marginal metrics are defined to quantify the rationality of each update.

\begin{definition}[\textbf{Marginal Privacy Gain, MPG}]
The reduction in adversarial inference accuracy achieved by the transformation at step $t$:
\begin{equation}
    \small
    \Delta P_t = P_{asr}(x^{(t-1)}) - P_{asr}(x^{(t)})
\end{equation}
where $x^{(t-1)}$ represents the intermediate text generated at the previous iteration.
\end{definition}

\begin{definition}[\textbf{Marginal Utility Cost, MUC}]
The semantic loss incurred at step $t$:
\begin{equation}
    \small
    \Delta C_t = U(x^{(0)}, x^{(t-1)}) - U(x^{(0)}, x^{(t)})
\end{equation}
\end{definition}

\begin{definition}[\textbf{Marginal Rate of Substitution, MRS}]
The instantaneous price of privacy preservation, representing the utility cost per unit of privacy gained:
\begin{equation}
\label{eq:mrs_def}
    \small
    \text{MRS}_t = \frac{\Delta C_t}{\Delta P_t}
\end{equation}
\end{definition}

From this perspective, a rational framework should ideally align with the principle that $\text{MRS}_t \le \lambda$, where $\lambda$ represents the maximum rational utility cost per unit of privacy.
Naive greedy strategies inherently tend to violate this condition and drift into irrational states of deadweight loss where $\text{MRS}_t \to \infty$.
Guided by this insight, RLAA is proposed to implicitly enforce this budget constraint $\lambda$ through architectural rationality, thereby structurally preventing utility collapse.
Detailed theoretical analysis is shown in Appendix \ref{app:theoretical_analysis}.

\begin{figure*}[t]
\centering
\includegraphics[width=\textwidth]{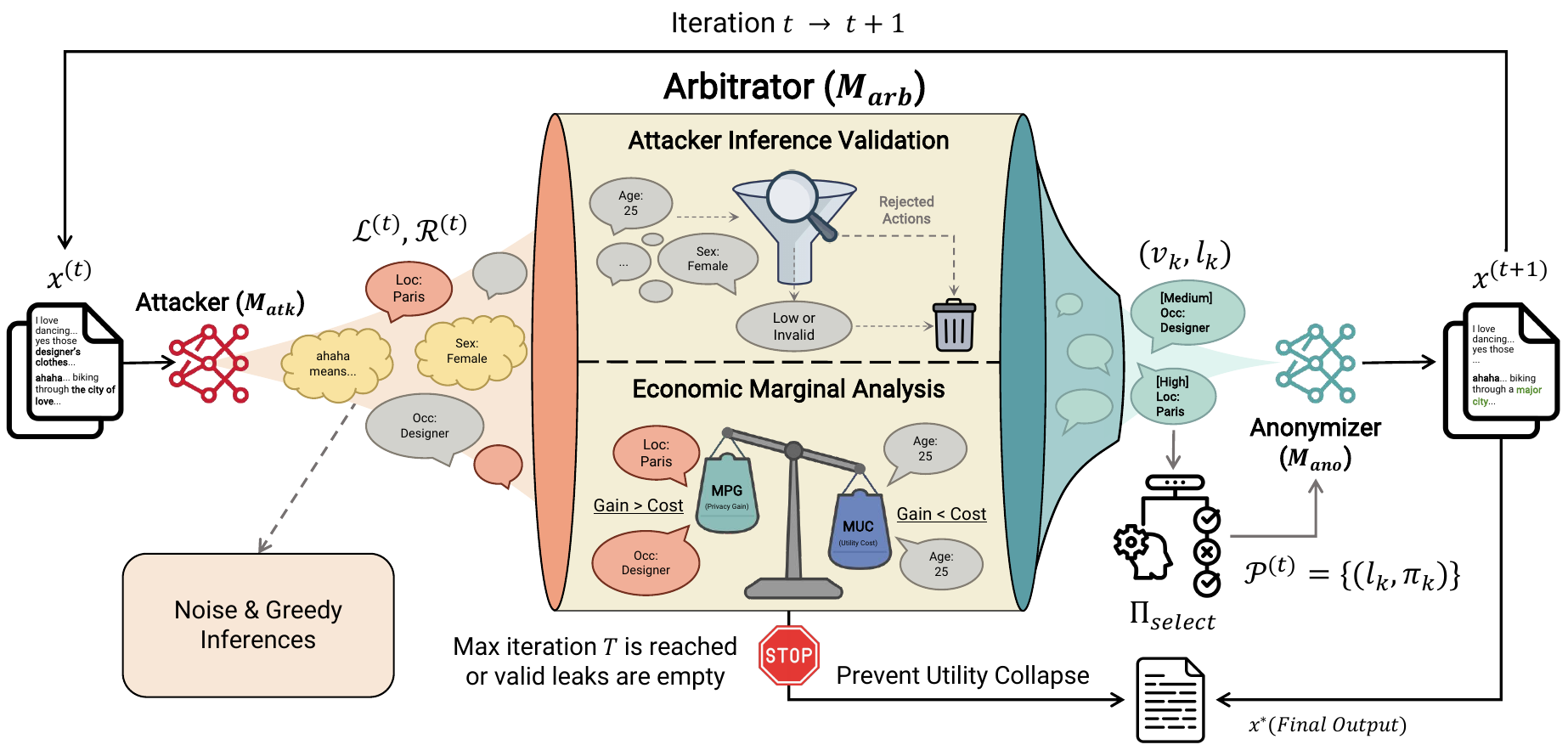}
\caption{\textbf{The RLAA Framework.} Utilizing an Attacker-Arbitrator-Anonymizer architecture, the arbitrator acts as a rationality gatekeeper. 
It validates attacker inferences to filter out ghost leaks with negligible privacy benefits, structurally preventing utility collapse caused by irrational greedy strategies.}
\label{fig:rlaa_framework}
\vspace{-0.2cm}
\end{figure*}

\subsection{The RLAA Framework}
\label{sec:rlaa_framework}
To address the problem mentioned above, we propose the Attacker-Arbitrator-Anonymizer (A-A-A) architecture illustrated in Figure \ref{fig:rlaa_framework}, which operates iteratively to refine the text $x^{(t)}$. 
For detailed algorithm settings and pseudo-code, please refer to Appendix \ref{subapp:algorithmic_procedure}.

\subsubsection{The Attacker ($M_{atk}$)}
The attacker acts as the sensory module. 
Given the current text $x^{(t)}$, it infers potential PII attributes and provides a reasoning chain:
\begin{equation}
    \small
    (\mathcal{L}^{(t)}, \mathcal{R}^{(t)}) = M_{atk}(x^{(t)})
\end{equation}
where $\mathcal{L}^{(t)}$ is the set of identified leaks and $\mathcal{R}^{(t)}$ is the set of corresponding inferences. 

\subsubsection{The Arbitrator ($M_{arb}$)}
The arbitrator functions as the central control module that regulates anonymization decisions considering the rationality constraints defined in Section \ref{sec:problem_formulation}. 
It includes a generative LSM backbone for logic validation and a deterministic control layer that parses validity signals and enforces the filtering process.
Instead of explicitly optimizing marginal trade-offs, which is unreliable given LSM's limited sensitivity to fine-grained scalar signals \cite{sun2025numerical}, the arbitrator validates attacker inferences by constraining the LSM to a structured discrimination task. 
Leveraging the greater reliability of verification over generation \cite{guan2024language}, this design enables self-correction and blocks ghost leaks from driving utility collapse.

For each leak $l_k \in \mathcal{L}^{(t)}$, the arbitrator assigns a validity level $v_k \in \mathcal{V}$, where the full validity space is $\mathcal{V} = \{\textsc{High}, \textsc{Med}, \textsc{Low}, \textsc{Invalid}\}$.
From the economic perspective defined in Section \ref{sec:problem_formulation}, this validity level acts as a discrete estimator of MPG. 
We logically partition $\mathcal{V}$ into the valid set $\mathcal{V}_{valid}$ (representing $\Delta P_t > 0$) and the ghost set $\mathcal{V}_{ghost}$ (representing $\Delta P_t \approx 0$), and the detailed partition is provided in Appendix \ref{subapp:algorithmic_procedure}.
Accordingly, we define the selection policy $\Pi_{select}$ as:
\begin{equation}
\small
\resizebox{0.8\linewidth}{!}{$
    \Pi_{select}(v_k) = 
    \begin{cases} 
    \textsc{Execute} & \text{if } v_k \in \mathcal{V}_{valid} \\
    \textsc{Ignore} & \text{if } v_k \in \mathcal{V}_{ghost}
    \end{cases}
    \label{eq:policy_select}
$}
\end{equation}
The \textsc{Execute} branch captures leakage with significant returns, ensuring a finite MRS.
Conversely, the \textsc{Ignore} branch filters out ghost leaks with $\Delta P_t \to 0$.
By rejecting these transactions, the arbitrator prevents the system from drifting into the MRS singularity identified in Eq. \ref{eq:mrs_def}, thereby structurally prevents utility collapse.

\subsubsection{The Anonymizer ($M_{ano}$)}
The anonymizer executes the refined policy $\mathcal{P}^{(t)}$ with the actionable leaks validated by the arbitrator:
\begin{equation}
\small
\resizebox{0.9\linewidth}{!}{$
    \mathcal{P}^{(t)} = \{ (l_k, \pi_k) \mid l_k \in \mathcal{L}^{(t)}, \Pi_{select}(v_k) = \textsc{EXECUTE} \}
$}
    \label{eq:policy_set}
\end{equation}
Based on this set, the text update is governed by:
\begin{equation}
\small
\resizebox{0.75\linewidth}{!}{$
    x^{(t+1)} = 
    \begin{cases} 
    M_{ano}(x^{(t)}, \mathcal{P}^{(t)}) & \text{if } \mathcal{P}^{(t)} \neq \emptyset \\
    x^{(t)} & \text{if } \mathcal{P}^{(t)} = \emptyset
    \end{cases}
$}
    \label{eq:anonymizer_update}
\end{equation}
When $\mathcal{P}^{(t)}$ is empty, the system triggers the early stop branch. 
This step guarantees the iteration converges to the fixed point $x^{(t+1)} = x^{(t)}$, thereby preventing the utility collapse caused by diminishing returns.

\section{Experiments}
\subsection{Experimental Setup}
\label{sec:exp_setup}
\noindent\textbf{Datasets.}
Two benchmark datasets are employed to evaluate our approach. 
\textbf{PersonalReddit} \cite{staab24beyond} is a  widely adopted dataset in research including IncogniText and RUPTA,
which consists of 525 human-verified synthetic Reddit conversations annotated with 8 fine-grained private attributes (\textit{e.g.}, age, gender and location). 
\textbf{reddit-self-disclosure} \cite{dou2024reducing} serves as a real-world evaluation bed, from which 885 samples containing self-disclosures of health conditions are extracted. 
Since this dataset is inherently labeled, we utilize the original labels as ground truth and perform manual verification to ensure consistency.

\noindent\textbf{Models.}
Our experiments involve two distinct model roles: a local model for the anonymization process and an external evaluation model for benchmarking privacy and utility.
We employ \textbf{Llama3-8B} and \textbf{Qwen2.5-7B} as local backbones, both of which facilitate consumer-grade deployment by requiring only approximately 4GB VRAM under 4-bit quantization.
To evaluate privacy, we employ the 685B \textbf{DeepSeek-V3.2-Exp} as a robust re-identification adversary ($\mathcal{A}_{re-id}$), following the threat model defined in Section \ref{sec:threat_model}.
Unlike proprietary APIs (\textit{e.g.}, GPT-4), which impose strict cost and access constraints, DeepSeek enables large-scale and strong attacks at negligible cost due to its open-source nature, making it particularly suitable for modeling an economically rational adversary operating under realistic resource constraints.
However, to mitigate the risk of self-referential evaluation from relying on a single backbone, we additionally re-evaluate the local methods using GPT-4o as an independent evaluator.
Detailed results are reported in Appendix \ref{app:gpt4o_eval}.

\noindent\textbf{Baseline Methods.}
We compared the performance of RLAA against \textbf{FgAA} \cite{staab2024large} (including its \textbf{Naive} and \textbf{SFT} migration variants), \textbf{SEAL} \cite{kim2025self}, \textbf{IncogniText} \cite{frikha2024incognitext} and \textbf{DP-BART-PR+} \cite{igamberdiev2023dp}.
The maximum adversarial iteration limit $T$ was configured according to dataset complexity and the need to evaluate varying convergence horizons: we set $T=10$ for the multi-attribute PersonalReddit and $T=3$ for the single-attribute reddit-self-disclosure.
Detailed configurations are provided in Appendix \ref{subapp:base_env} and \ref{subapp:gen_params}.

\subsection{Evaluation Metrics}
\label{sec:metrics}
To evaluate the performance of all methods, we adopt the standardized privacy and utility evaluation protocols widely established in recent LLM-based anonymization research \cite{staab2024large, frikha2024incognitext, kim2025self}. 
Let $D = \{(x_i^{ori}, A_i)\}_{i=1}^N$ be the test dataset, where $x_i^{ori}$ is the original text and $A_i$ is the set of ground-truth private attributes. 
Let $x_i^{ano}$ be the anonymized text generated by a given method.

\noindent\textbf{Privacy (PRIV).}
The performance on privacy preservation is evaluated using the anonymization under LLM inference setting \cite{staab2024large}. 
We use a powerful adversary model $\mathcal{M}_{atk}$ (DeepSeek-V3.2-Exp) to infer the set of attributes $A'_i = \mathcal{M}_{atk}(x_i^{ano})$ from the anonymized text. 
The score is the average attack success rate over all attributes $K$ across all $N$ samples, which employs a programmatic matcher detailed in Appendix \ref{subapp:privacy_eval}. 
\begin{equation}
    \small
    \text{PRIV} = \frac{1}{N \cdot K} \sum_{i=1}^N \sum_{k=1}^K \mathbb{I}(A'_{i,k} \approx A_{i,k}) 
\end{equation}

\begin{table*}[t!]
\centering
\small
\renewcommand{\arraystretch}{1.15}
\setlength{\tabcolsep}{1.5pt}
\resizebox{\textwidth}{!}{
\begin{tabular}{ c @{\hspace{8pt}} l c YYYYYYY }
\toprule
\multicolumn{2}{c}{\textbf{Method}} & \textbf{Base Model} & 
\textbf{UTIL} $\uparrow$ & \textbf{PRIV} $\downarrow$ & \textbf{ROUGE} $\uparrow$ & \textbf{BLEU} $\uparrow$ & \textbf{MEAN} $\uparrow$ & \textbf{READ} $\uparrow$ & \textbf{HALL} $\uparrow$ \\
\midrule
\multicolumn{10}{c}{\textbf{\texttt{PersonalReddit}}} \\
\midrule
\multicolumn{2}{c}{Original Text} & - & 
1.0000 & 0.4442 & 1.0000 & 1.0000 & 10.000 & 10.000 & 1.0000 \\
\midrule
\multirow{4}{*}{\rotatebox[origin=c]{90}{\scriptsize \shortstack{\textbf{Local}\\\textbf{Methods}}}} 
 & DP-BART+ & BART-Base & 0.3470 & 0.2650 & \underline{0.3925} & \underline{0.2597} & 2.6610 & 7.5790 & 0.0170 \\
 & IncogniText & Llama3-8B & 0.6330 & \textbf{0.1230} & 0.3499 & 0.2304 & \underline{5.5040} & \textbf{10.0000} & 0.3470 \\
 & FgAA-Naive & Llama3-8B & \underline{0.7297} & \underline{0.1948} & 0.2180 & 0.0533 & 3.5420 & 8.9330 & \underline{0.9420} \\
 & \textbf{RLAA} & Llama3-8B & \textbf{0.8788} & 0.2130 & \textbf{0.5958} & \textbf{0.4251} & \textbf{7.0780} & \underline{9.8090} & \textbf{0.9480} \\
\midrule
\multirow{3}{*}{\rotatebox[origin=c]{90}{\scriptsize \shortstack{\textbf{API}\\\textbf{Required}}}}
 & SEAL & Llama3-8B & 0.4642 & 0.1787 & 0.1205 & 0.0753 & 9.2645 & 1.5207 & 0.3140 \\
 & FgAA-SFT & Llama3-8B & 0.9670 & 0.2940 & 0.9149 & 0.9389 & 9.2310 & 9.9340 & 0.9830 \\
 & FgAA-API & DeepSeek-V3.2-Exp & 0.8264 & 0.2056 & 0.4649 & 0.2082 & 5.4380 & 9.3550 & 1.0000 \\
\midrule
\multicolumn{10}{c}{\textbf{\texttt{reddit-self-disclosure}}} \\
\midrule
\multicolumn{2}{c}{Original Text} & - & 
1.0000 & 0.4943 & 1.0000 & 1.0000 & 10.000 & 10.000 & 1.0000 \\
\midrule
\multirow{4}{*}{\rotatebox[origin=c]{90}{\scriptsize \shortstack{\textbf{Local}\\\textbf{Methods}}}} 
 & DP-BART+ & BART-Base & 0.3999 & 0.3245 & 0.2565 & 0.1093 & 2.3698 & 8.6830 & 0.0943 \\
 & IncogniText & Llama3-8B & 0.7755 & \underline{0.1283} & \textbf{0.8584} & \textbf{0.8781} & \textbf{7.0792} & \textbf{9.9585} & 0.6226 \\
 & FgAA-Naive & Llama3-8B & \underline{0.8187} & 0.1591 & 0.4919 & 0.2805 & 5.3523 & 9.6629 & \textbf{0.9545} \\
 & \textbf{RLAA} & Llama3-8B & \textbf{0.8572} & \textbf{0.1136} & \underline{0.6016} & \underline{0.4703} & \underline{6.8939} & \underline{9.6932} & \underline{0.9129} \\
\midrule
\multirow{3}{*}{\rotatebox[origin=c]{90}{\scriptsize \shortstack{\textbf{API}\\\textbf{Required}}}}
 & SEAL & Llama3-8B & 0.6303 & 0.0226 & 0.1213 & 0.0736 & 9.8377 & 2.2415 & 0.6830 \\
 & FgAA-SFT & Llama3-8B & 0.9218 & 0.1925 & 0.7813 & 0.7333 & 8.0528 & 9.9019 & 0.9698 \\
 & FgAA-API & DeepSeek-V3.2-Exp & 0.9118 & 0.1660 & 0.7573 & 0.6626 & 7.5210 & 9.8720 & 0.9960 \\
\bottomrule
\end{tabular}
}
\caption{\textbf{Main Baseline Comparison.} Performance of RLAA against local anonymization baselines (DP-BART+, IncogniText and FgAA-Naive) and baselines that require external API access during training, supervision, or some stage of the pipeline (SEAL, FgAA-SFT and FgAA-API).}
\label{tab:comprehensive_baseline_comparison}
\end{table*}

\begin{figure*}[t!]
\centering
\includegraphics[width=0.75\textwidth]{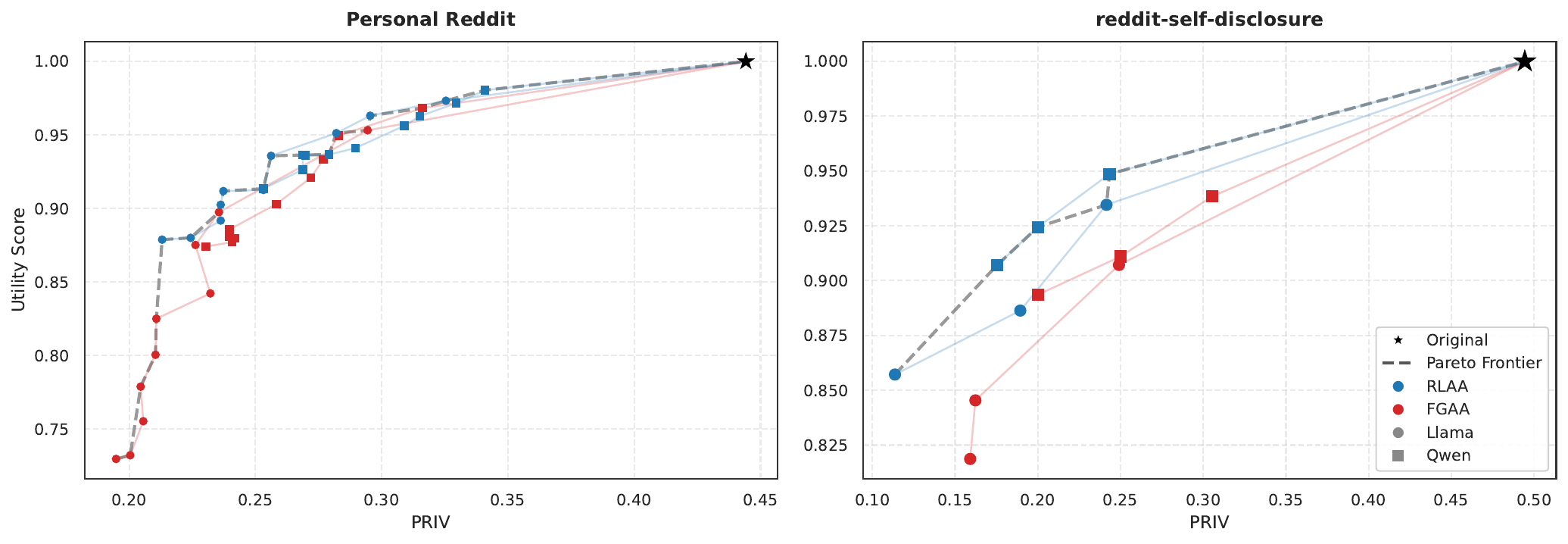}
\caption{\textbf{Privacy-Utility Trade-off.} RLAA achieves superior trade-offs compared to FgAA across iterations on two datasets. The trade-off dynamics for structural metrics (ROUGE/BLEU) are detailed in Appendix \ref{app:traditional_trade_off}.}
\label{fig:util_trade_off}
\vspace{-0.2cm}
\end{figure*}

\begin{table*}[t!]
\small
\centering
\setlength{\tabcolsep}{3pt}
\begin{tabular}{l c Y Y Y Y Y Y Y} 
\toprule
\textbf{Base Model} & \textbf{Method} & \textbf{UTIL} $\uparrow$ & \textbf{PRIV} $\downarrow$ & \textbf{ROUGE} $\uparrow$ & \textbf{BLEU} $\uparrow$ & \textbf{MEAN} $\uparrow$ & \textbf{READ} $\uparrow$ & \textbf{HALL} $\uparrow$ \\
\midrule
\multicolumn{9}{c}{\textbf{\texttt{PersonalReddit}}} \\
\midrule
\multicolumn{1}{l}{Original Text} & - & 1.0000 & 0.4442 & 1.0000 & 1.0000 & 10.0000 & 10.0000 & 1.0000 \\
\cmidrule{1-9}
\textbf{DeepSeek-V3.2-Exp}
& w/o Arb. & 0.8264 & \textbf{0.2056} & 0.4649 & 0.2082 & 5.4380 & 9.3550 & \textbf{1.0000} \\
& \textbf{RLAA} & \textbf{0.9240} & 0.2087 & \textbf{0.7401} & \textbf{0.6365} & \textbf{7.8350} & \textbf{9.9670} & 0.9920 \\
\cmidrule{1-9}
\textbf{Llama3-8B}
& w/o Arb. & 0.7297 & \textbf{0.1948} & 0.2180 & 0.0533 & 3.5420 & 8.9330 & 0.9420 \\
& \textbf{RLAA} & \textbf{0.8788} & 0.2130 & \textbf{0.5958} & \textbf{0.4251} & \textbf{7.0780} & \textbf{9.8090} & \textbf{0.9480} \\
\cmidrule{1-9}
\textbf{Qwen2.5-7B}
& w/o Arb. & 0.8741 & \textbf{0.2302} & 0.6156 & 0.3973 & 6.8590 & 9.6120 & \textbf{0.9750} \\
& \textbf{RLAA} & \textbf{0.9135} & 0.2531 & \textbf{0.7549} & \textbf{0.6413} & \textbf{7.8020} & \textbf{9.8620} & 0.9740 \\
\midrule
\multicolumn{9}{c}{\textbf{\texttt{reddit-self-disclosure}}} \\
\midrule
\multicolumn{1}{l}{Original Text} & - & 1.0000 & 0.4943 & 1.0000 & 1.0000 & 10.0000 & 10.0000 & 1.0000 \\
\cmidrule{1-9}
\textbf{DeepSeek-V3.2-Exp}
& w/o Arb. & 0.9118 & 0.1660 & 0.7573 & 0.6626 & 7.5210 & 9.8720 & \textbf{0.9960} \\
& \textbf{RLAA} & \textbf{0.9132} & \textbf{0.1434} & \textbf{0.8070} & \textbf{0.7364} & \textbf{7.5620} & \textbf{9.9090} & 0.9920 \\
\cmidrule{1-9}
\textbf{Llama3-8B}
& w/o Arb. & 0.8187 & 0.1591 & 0.4919 & 0.2805 & 5.3523 & 9.6629 & \textbf{0.9545} \\
& \textbf{RLAA} & \textbf{0.8572} & \textbf{0.1136} & \textbf{0.6016} & \textbf{0.4703} & \textbf{6.8939} & \textbf{9.6932} & 0.9129 \\
\cmidrule{1-9}
\textbf{Qwen2.5-7B}
& w/o Arb. & 0.8936 & 0.2000 & 0.6458 & 0.4800 & 7.2115 & 9.8269 & \textbf{0.9769} \\
& \textbf{RLAA} & \textbf{0.9071} & \textbf{0.1753} & \textbf{0.7273} & \textbf{0.6208} & \textbf{7.7200} & \textbf{9.9320} & 0.9560 \\
\bottomrule
\end{tabular}
\caption{\textbf{Ablation and Generalization}. Comparison between the full RLAA framework and the baseline without arbitrator. The arbitrator consistently improves utility across different model scales and datasets.}
\label{tab:pairwise_comparison}
\end{table*}

\begin{figure*}[t!]
    \centering
    \includegraphics[width=0.49\linewidth]{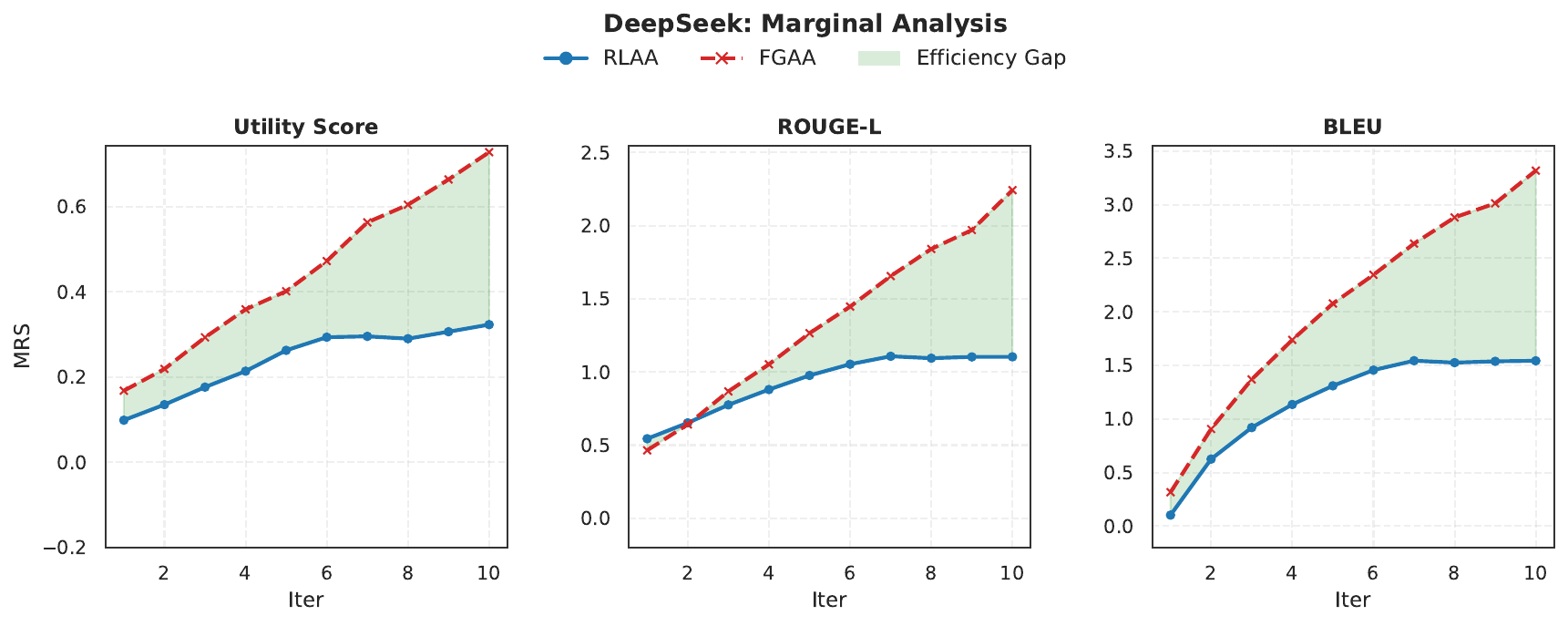}
    \includegraphics[width=0.49\linewidth]{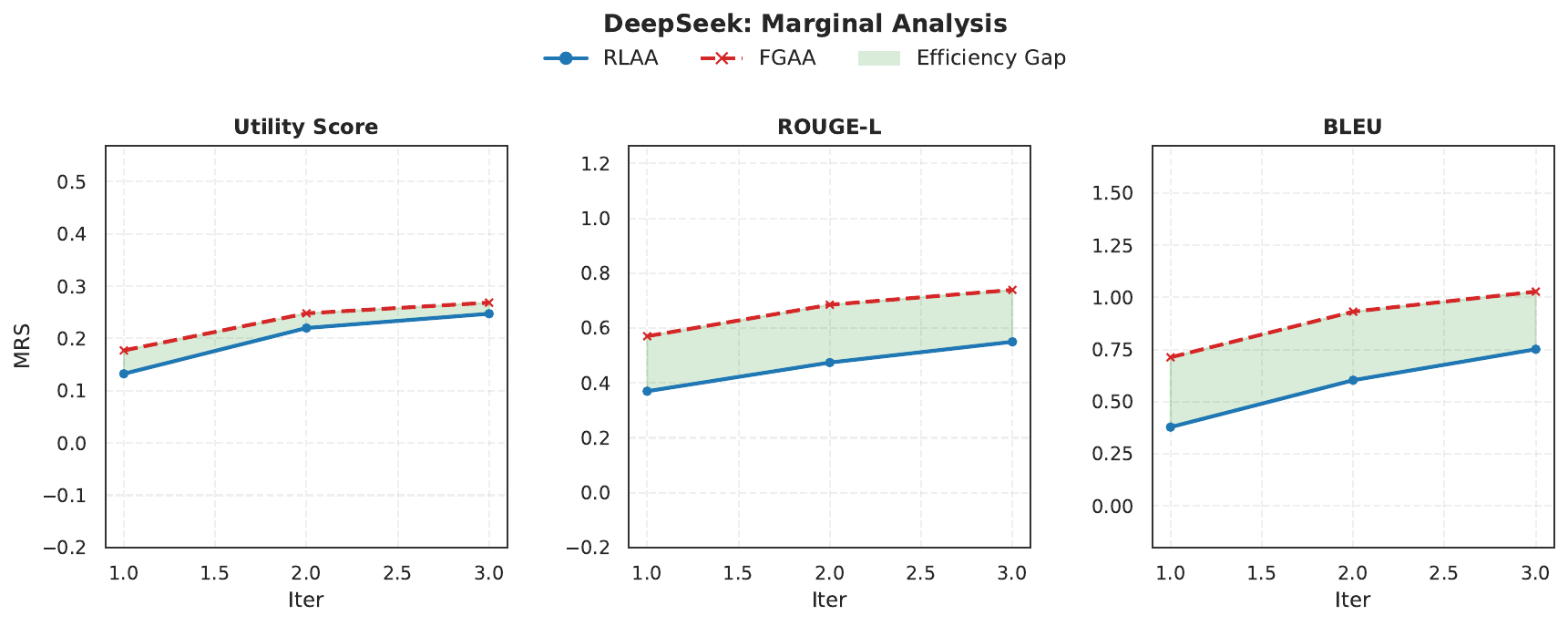}
    \includegraphics[width=0.49\linewidth]{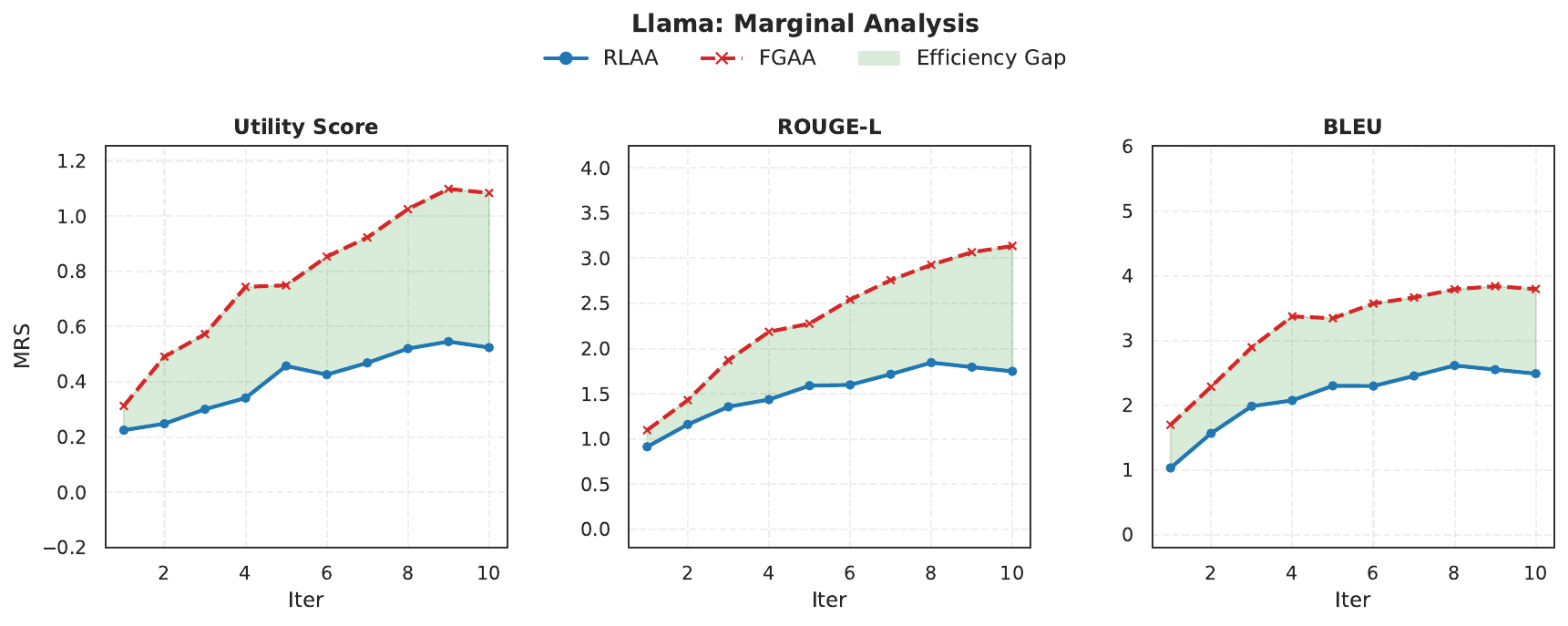}
    \includegraphics[width=0.49\linewidth]{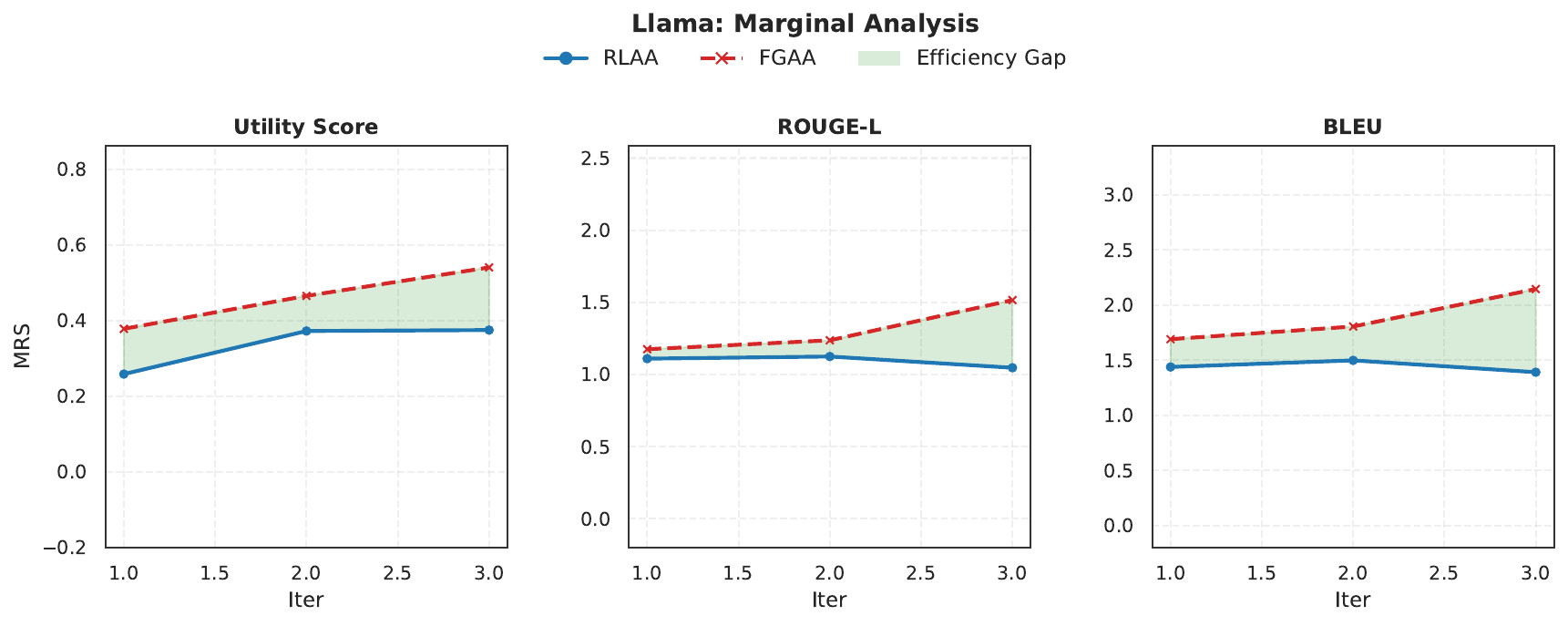}
    \includegraphics[width=0.49\linewidth]{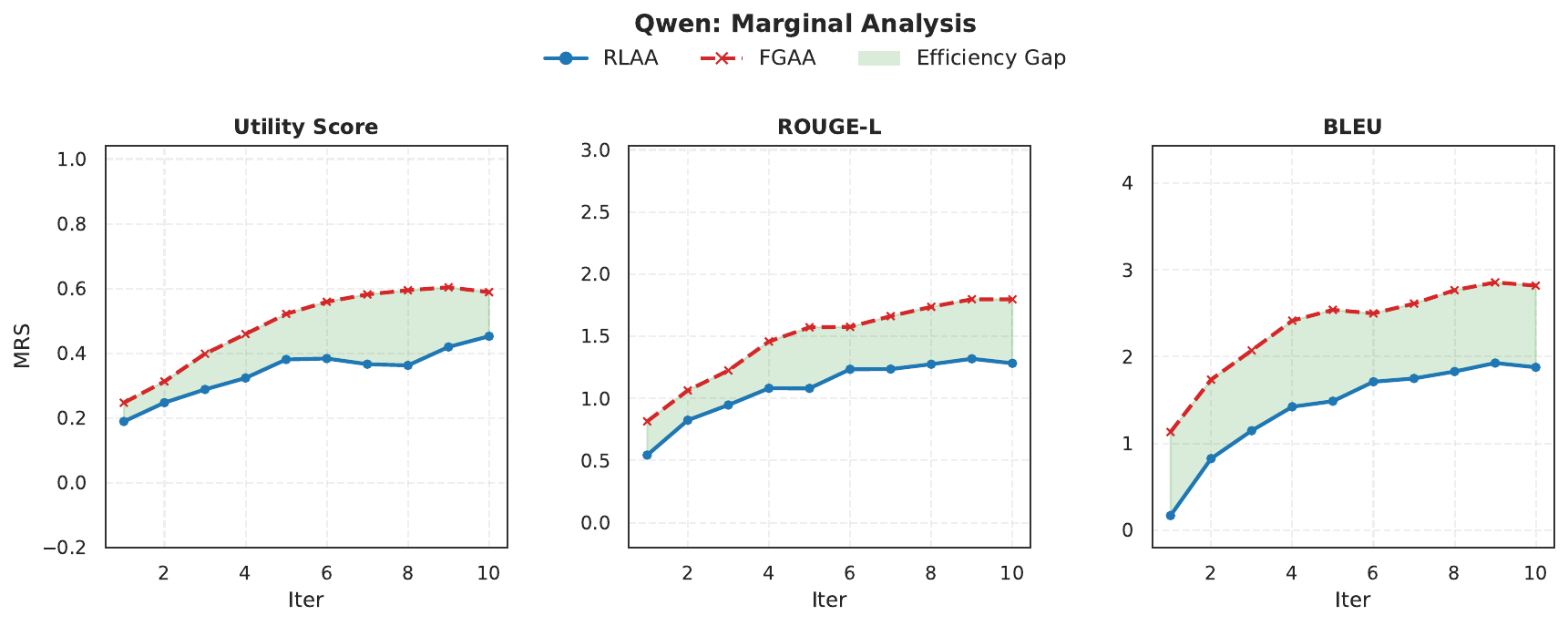}
    \includegraphics[width=0.49\linewidth]{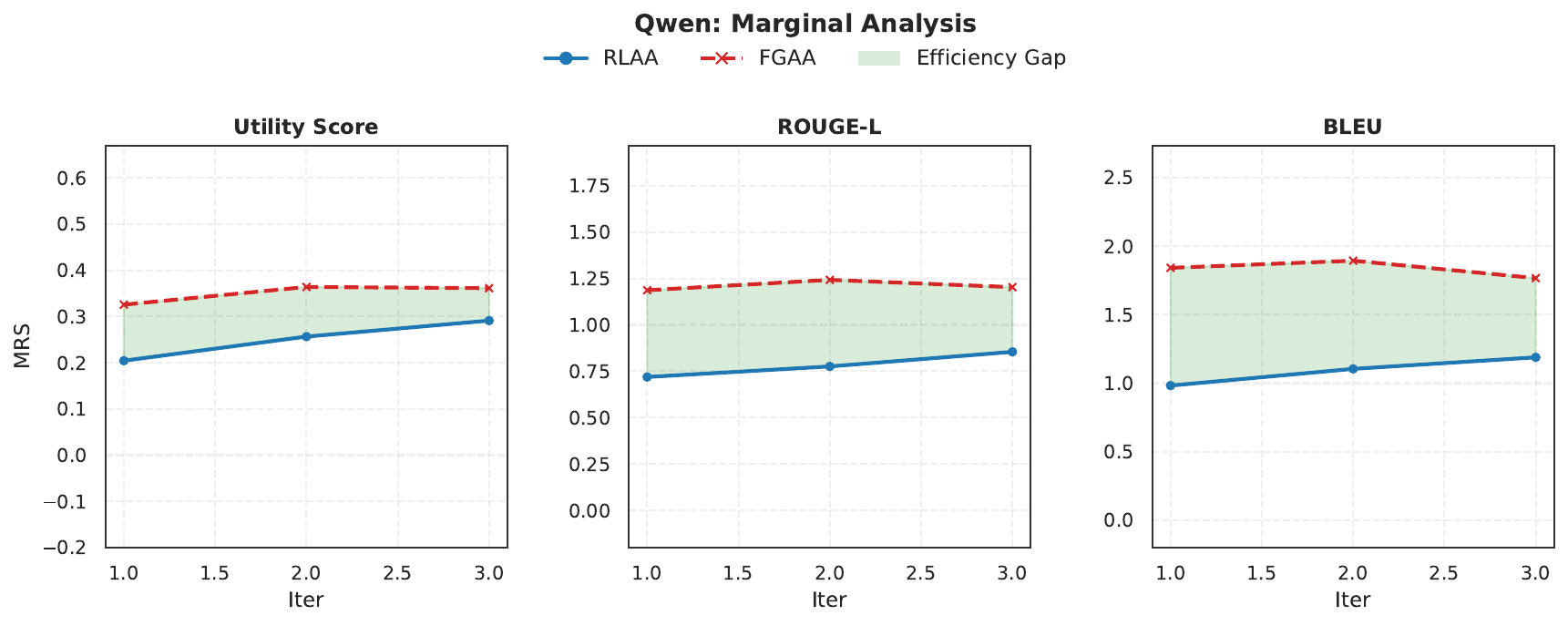}
    \caption{\textbf{Cumulative MRS Analysis.} 
    The figure displays the cumulative MRS during the anonymization process on PersonalReddit (Left) and reddit-self-disclosure (Right). FgAA (Red) shows a sustained increase while RLAA (Blue) maintains a stable low MRS. Detailed quantitative economic analysis is provided in Appendix \ref{subapp:quantitative_analysis}.}
    \label{fig:mrs_analysis}
    \vspace{-0.2cm}
\end{figure*}

\noindent\textbf{Utility (UTIL).}
The utility evaluation relies on an LLM-as-a-Judge $\mathcal{M}_{judge}$ (DeepSeek-V3.2-Exp) \cite{zhang2024small} to assess semantic preservation via $K$ distinct components. 
We denote the scalar score range as $[S_{\min}, S_{\max}]$:
\noindent \textbf{(1) Readability ($s_{\text{read}}$)} is a scalar score $\in [S_{\min}, S_{\max}]$ assessing if the text is understandable on its own, defined as $s_{\text{read},i} = \mathcal{M}_{judge}(x_i^{ano})$.
\noindent \textbf{(2) Meaning ($s_{\text{mean}}$)} is a scalar score $\in [S_{\min}, S_{\max}]$ assessing if the core message is preserved, where $s_{\text{mean},i} = \mathcal{M}_{judge}(x_i^{ori}, x_i^{ano})$.
\noindent \textbf{(3) Hallucinations ($s_{\text{hall}}$)} is a binary score $s_{\text{hall},i} = \mathbb{I}(\text{no new ungrounded info})$
The final score is the average of these $K$ normalized components ($K=3, S_{\min}=1, S_{\max}=10$):
\begin{equation}
\label{eq:util}
    \small
    \text{UTIL} = \frac{1}{K \cdot N} \sum_{i=1}^N \left( \frac{s_{\text{read},i}}{S_{\max}} + \frac{s_{\text{mean},i}}{S_{\max}} + s_{\text{hall},i} \right)
\end{equation}

\noindent\textbf{Traditional Structural Metrics.}
In addition to LLM-based \textbf{UTIL Score}, we report standard n-gram metrics to evaluate structural similarity:
\textbf{ROUGE} denotes the ROUGE-L F1 Score \citep{lin2004rouge} based on the longest common subsequence, and \textbf{BLEU} \citep{papineni2002bleu} indicates the bilingual evaluation understudy score for precision.

\subsection{Performance Analysis}
\label{sec:performance_analysis}
Our analysis begins by establishing RLAA's empirical superiority against strong baselines in Section \ref{sec:baseline_compare} and verifying the arbitrator's crucial role through ablation studies in Section \ref{sec:ablation_study}. 
The empirical findings are then bridged with our theoretical model via an economic analysis in Section \ref{sec:economic_analysis}.
Finally, a human study in addition to automatic evaluations is presented in Section \ref{sec:human_eval}.

\subsubsection{Baseline Comparison Results}
\label{sec:baseline_compare}
Table \ref{tab:comprehensive_baseline_comparison} presents the comparative analysis on the PersonalReddit and reddit-self-disclosure datasets.
RLAA achieves the optimal privacy-utility balance among local methods (UTIL=0.8788/0.8572).
In contrast, FgAA-Naive experiences a pronounced utility collapse due to greedy over-editing. 
Figure \ref{fig:util_trade_off} visualizes this performance gap across different horizons: FgAA exhibits progressive utility degradation consistent with diminishing returns on PersonalReddit, while RLAA maintains a superior Pareto frontier from the outset on reddit-self-disclosure.
Besides, IncogniText sacrifices significant semantics (HALL=0.3470) by fabricating attributes that introduce ungrounded and potentially misleading contents.
Distillation-based methods reveal a critical failure mode,
\textit{i.e.},
SEAL achieves strong privacy (PRIV=0.1787/0.0226) but catastrophic utility loss 
(UTIL=0.4642/0.6303). 
We argue this stems from distilling multi-round teacher 
trajectories without early-stopping signals: the student inherits 
privacy-seeking aggressiveness but never learns to stop. 
Conversely, our FgAA-SFT variant (detailed in Appendix 
\ref{subapp:training_convergence}) corrects for over-caution (PRIV=0.2940/0.1925). 
It reveals distillation's fundamental limitation: it transfers editing behaviors but not the meta-judgment of rational stopping. In contrast, RLAA's training-free arbitrator structurally enforces this rationality, even rivaling FgAA-API (DeepSeek-685B) with 8B models.

\subsubsection{Ablation Study}
\label{sec:ablation_study}
To isolate the arbitrator's contribution, we treat FgAA as the ablation 
baseline without arbitrator \textbf{w/o Arb.} 
Table \ref{tab:pairwise_comparison} shows consistent utility gains across all models. 
Crucially, RLAA achieves Pareto dominance on reddit-self-disclosure, improving both privacy (0.1591→0.1136) and utility (0.8187→0.8572), which demonstrates that rationality constraints actively optimize rather than merely constrain.
Even the powerful DeepSeek-685B benefits from the arbitrator, 
which confirms again the irrationality also stems from the greedy strategy itself. 
This observation validates RLAA as a general component to compensate for the rationality of model behaviors, regardless of scales.

\subsubsection{Economic Efficiency Analysis}
\label{sec:economic_analysis}
To validate our theoretical premise regarding rationality, we analyze the marginal dynamics of the anonymization process.
Figure \ref{fig:mrs_analysis} visualizes the MRS across iterations for Utility Score.
FgAA (Red Curve) exhibits a continuous increase in MRS across both datasets, which confirms that the greedy strategy tends to drift into an economically inefficient state.
In contrast, RLAA (Blue Curve) maintains a low and stable MRS trajectory, reducing the terminal MRS based on UTIL Score from \textbf{3.80} to \textbf{1.74} on Llama3-8B.
The comparison serves as an empirical proof of our hypothesis: The constraint on rationality prevents utility collapse.
Besides, further quantitative analysis in Appendix \ref{subapp:quantitative_analysis} reveals a counter-intuitive capability-rationality paradox: SOTA level models exhibit a steeper irrationality drift and achieve higher rationality gains than smaller models.
This observation suggests that scaling capabilities alone cannot resolve rationality alignment failures.

\subsubsection{Human Evaluation}
\label{sec:human_eval}
To further examine the semantic utility of RLAA's generated text and the arbitrator's reliability, we conduct two complementary human evaluations on PersonalReddit dataset.

\begin{figure}[h]
    \centering
    \includegraphics[width=0.45\textwidth]{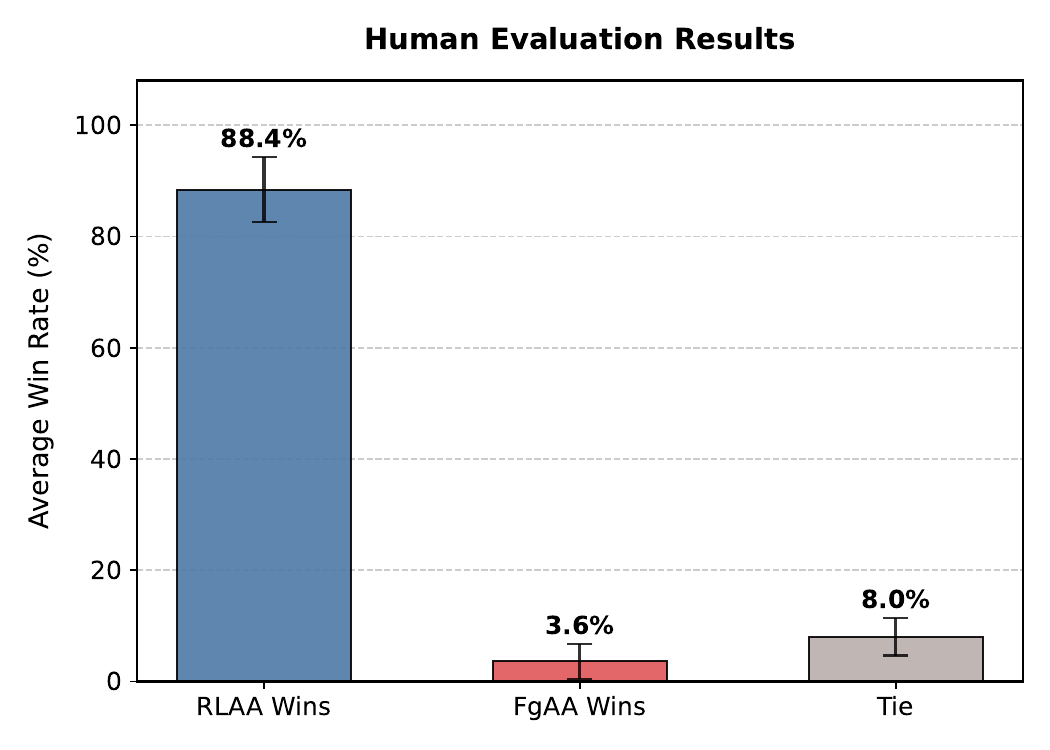}
    \caption{\textbf{Human pairwise evaluation results.} RLAA achieves a dominant win rate against FgAA-Naive while reflecting high inter-annotator consistency.}
    \label{fig:human_eval}
\end{figure}

First, to directly assess semantic utility, we conduct a blind three-way pairwise comparison between the outputs of RLAA, FgAA-Naive and the original text.
The evaluation involves five independent non-author annotators, each judging 50 randomized samples, resulting in 250 total judgments.
As shown in Figure \ref{fig:human_eval}, RLAA achieved a dominant average win rate of \textbf{88.4\%}, while the FgAA's win rate was only \textbf{3.6\%}. 
RLAA's decisive preference rating substantiates that the arbitrator effectively preserves semantic integrity.
We also observed a strong consensus across all five reviewers with a remarkably low standard deviation in win rates ($\sigma=5.9\%$), which provides strong statistical evidence for inter-annotator consistency and ensures the objective reliability of these findings.

Second, to evaluate whether the arbitrator can reliably distinguish genuine leaks from ghost leaks, we conduct a human validation study on PersonalReddit with three annotators with NLP research backgrounds.
The annotated set contains 200 randomly sampled instances from anonymization trajectories for overall reliability evaluation and 240 additional samples stratified across validity tiers for tier-based calibration; majority vote is used as the final label.
Table \ref{tab:arb_reliability} shows that the arbitrator achieves high recall and relatively low false-negative rates across both Llama3-8B and Qwen2.5-7B, indicating that genuine leaks are rarely filtered out by early stopping.
To further assess calibration, we compute the proportion of human-confirmed genuine leaks within each validity tier.
As shown in Table \ref{tab:arb_tier_calibration}, this proportion decreases monotonically from \texttt{HIGH} to \texttt{INVALID} for both local backbones, supporting the interpretation that validity tiers provide an empirically grounded proxy for actionable privacy severity.

\begin{table}[t!]
\centering
\small
\renewcommand{\arraystretch}{1.05}
\setlength{\tabcolsep}{5pt}
\begin{tabular}{lcc}
\toprule
\textbf{Metric} & \textbf{Llama3-8B} & \textbf{Qwen2.5-7B} \\
\midrule
\textbf{Recall}    & 95.45\% & 85.71\% \\
\textbf{Precision} & 58.33\% & 60.00\% \\
\textbf{Accuracy}  & 84.00\% & 85.00\% \\
\textbf{FNR}       & 4.55\%  & 14.29\% \\
\textbf{FPR}       & 19.23\% & 15.19\% \\
\textbf{F1}        & 0.7241  & 0.7059 \\
\bottomrule
\end{tabular}
\caption{\textbf{Human-validated arbitrator reliability.} FNR and FPR denote false negative and false positive rates.}
\label{tab:arb_reliability}
\end{table}

\begin{table}[t!]
\centering
\small
\renewcommand{\arraystretch}{1.05}
\setlength{\tabcolsep}{5pt}
\begin{tabular}{lcc}
\toprule
\textbf{Tier} & \textbf{Llama3-8B} & \textbf{Qwen2.5-7B} \\
\midrule
\texttt{HIGH}    & 90.00\% & 90.00\% \\
\texttt{MED}     & 36.67\% & 53.33\% \\
\texttt{LOW}     & 20.00\% & 10.00\% \\
\texttt{INVALID} & 3.33\%  & 3.33\% \\
\bottomrule
\end{tabular}
\caption{\textbf{Tier-based calibration.} Entries are the proportion of human-confirmed genuine leaks in each tier.}
\label{tab:arb_tier_calibration}
\end{table}

\section{Conclusion}
\label{sec:conclusion}
This paper reframes utility collapse in localized adversarial anonymization as a failure of rational decision-making and proposes RLAA, a training-free framework built on an Attacker-Arbitrator-Anonymizer architecture. 
By validating attacker inferences and introducing verification-guided early stopping, RLAA prevents semantic utility collapse caused by hallucinations and diminishing returns. 
Extensive experiments show that RLAA consistently achieves a stronger privacy-utility trade-off than competitive baselines. 
Grounding safety in economic rationality, RLAA resolves the privacy paradox, aligning safety with model capability.

\clearpage

\section*{Limitations}
\label{sec:limitations}
While RLAA achieves superior performance in extensive experiments, there are several limitations to acknowledge and to be mitigated in future work:

\noindent\textbf{Computational Overhead.}
The current RLAA operates as an inference-time alignment mechanism, where the arbitrator enforces a slow-thinking verification pass. 
It is necessary for rationality but introduces additional computational overhead: 
A detailed quantitative analysis is provided in Appendix \ref{app:overhead_analysis}, which indicates that this overhead remains acceptable for most quality-centric offline scenarios. 
Future work could leverage RLAA to generate high-quality and rationality-aligned trajectories.
By fine-tuning models on these trajectories, we can internalize the arbitrator’s external constraints into the model’s parameters to achieve training-time alignment, thereby eliminating the computational overhead of an external gatekeeper.

\noindent\textbf{Evolving Adversarial Capabilities.}
Our evaluation relies on DeepSeek-V3.2-Exp to simulate a powerful adversary. 
While this represents a current SOTA threat model, the landscape of LLM attacks is rapidly evolving.
RLAA provides empirical defense, but we cannot theoretically guarantee immunity against future models with significantly stronger reasoning capabilities or attackers possessing extensive auxiliary background knowledge.

\noindent\textbf{Lack of Provable Guarantees.}
Unlike Differential Privacy (DP) mechanisms, RLAA does not provide mathematically provable privacy budgets ($\epsilon$).
This is a known trade-off in LLM-based text anonymization: DP methods often result in significant degradation of semantic readability and utility.
Our approach prioritizes semantic preservation and empirical safety, positioning it as a pragmatic solution rather than a mathematically guaranteed one.

\section*{Ethics Statement}
We acknowledge the dual-use risks of LLMs and employ our adversarial framework strictly for defensive evaluation. 
Our approach enhances privacy by enabling fully localized deployment, thereby eliminating data exposure to third-party APIs. 
We clarify that RLAA provides empirical defense against SOTA adversaries, rather than provable guarantees like Differential Privacy. 
Regarding data usage, we rely exclusively on established public datasets (PersonalReddit and reddit-self-disclosure) and strictly adhere to ethical standards by avoiding any new collection of private data.

\bibliography{anthology,custom}

\appendix
\section*{Appendix}
\section{Theoretical Analysis}
\label{app:theoretical_analysis}
In this section, we provide a formal economic definition and analytical derivation showing that adversarial strategies like FgAA are inherently irrational, while RLAA achieves structural rationality.

\begin{definition}[\textbf{Economic Rationality Condition}]
An anonymization framework is defined as economically rational only if every executed step satisfies the budget constraint:
\begin{equation}
    \text{MRS}_t \le \lambda
\end{equation}
where $\lambda \in \mathbb{R}^+$ represents the maximum acceptable utility cost per unit of privacy gain. 
$\lambda$ is an intrinsic property derived from the specific utility metric $U(\cdot)$ and deployment constraints, representing the break-even point where the marginal cost of anonymization outweighs its privacy benefit.
\end{definition}

\begin{assumption}[\textbf{Discrete Estimation}]
We model the arbitrator's discrete validity judgment as a quantized estimator of the expected privacy gain $\mathbb{E}[\Delta P]$.
Using the sets defined in Section \ref{sec:rlaa_framework} ($\mathcal{V}_{valid}$ and $\mathcal{V}_{ghost}$), we map discrete labels to expected gains:
\begin{equation}
    \small
    \mathbb{E}[\Delta P_t(l)] \approx \begin{cases} 
    \gamma_{v} > 0 & \text{if } v \in \mathcal{V}_{valid} \\
    \gamma_{g} \approx 0 & \text{if } v \in \mathcal{V}_{ghost}
    \end{cases}
\end{equation}
where $\gamma_{v}$ represents a significant privacy improvement, and $\gamma_{g}$ represents the negligible gain characteristic of a ghost leak.
\end{assumption}

\begin{proposition}[\textbf{Implicit Budget Enforcement}]
The arbitrator's discrete mechanism creates a binary decision boundary that is functionally equivalent to an economic mechanism operating with an implicit budget $\lambda$.
\end{proposition}

\begin{proof}[Derivation]
Assuming that the utility cost for any atomic edit is lower-bounded by $\Delta C_t \ge \epsilon$, we define the MRS lower bound for actionable leaks and the MRS upper bound for ghost leaks:
\begin{equation}
    \text{MRS}_{valid} \approx \frac{\epsilon}{\gamma_{v}}, \quad \text{MRS}_{ghost} \approx \frac{\epsilon}{\gamma_{g}} \to \infty
\end{equation}
Under the assumption of correct estimation, there exists a significant separation gap: $\text{MRS}_{valid} \ll \text{MRS}_{ghost}$.
And the arbitrator's policy $\Pi_{select}$ effectively implements a budget constraint $\lambda$ located within this gap:
\begin{equation}
\small
\begin{split}
    \exists \lambda \in \left[ \frac{\epsilon}{\gamma_{v}}, \frac{\epsilon}{\gamma_{g}} \right) \\
    \text{ s.t. } \forall v, \Pi_{select}(v) \neq \textsc{Ignore} & \iff \text{MRS}(v) \le \lambda
\end{split}
\label{eq:budget_existence}
\end{equation}
This mechanism ensures the system structurally prioritizes high-return transactions while rejecting deadweight losses caused by ghost leaks.
\end{proof}

\begin{corollary}[\textbf{Hallucination Defense - Instantaneous Rationality}]
In a greedy system, hallucinated ghost leaks $l_{hall}$ incur definitive utility costs ($\epsilon$) for negligible privacy gains ($\gamma_{g}$). We can formally derive the irrationality of executing such edits as a singularity in the MRS:
\begin{equation}
    \lim_{\Delta P \to \gamma_{g}} \text{MRS}(l_{hall}) \approx \frac{\epsilon}{\gamma_{g}} \to \infty
\end{equation}
By applying \textbf{Proposition 1}, RLAA identifies such instances as belonging to the ghost set ($v \in \mathcal{V}_{ghost}$). 
Since the implicit MRS exceeds any rational budget $\lambda$, the policy triggers a rejection $\Pi_{select} = \textsc{Ignore}$. 
Thus, the transaction cost is forced to zero ($\Delta C_t = 0$), thereby avoiding the immediate deadweight loss.
\end{corollary}

\begin{corollary}[\textbf{Diminishing Returns - Asymptotic Rationality}]
As the iteration $t$ increases, the process enters a long-tail phase where remaining candidates are exclusively ghost leaks with $\lim_{t \to \infty} \Delta P_t = \gamma_{g}$.
A greedy strategy fails to stop because it lacks a mechanism to evaluate the diverging cost-benefit ratio:
\begin{equation}
    \lim_{t \to \infty} \text{MRS}_{greedy} = \frac{\Delta C_t}{\lim_{t \to \infty} \Delta P_t} \to \infty
\end{equation}
In contrast, RLAA applies \textbf{Proposition 1} to the entire set of remaining candidates $\mathcal{L}^{(t)}$. Since $\forall l \in \mathcal{L}^{(t)}, l \in \mathcal{V}_{ghost} \implies \text{MRS}(l) > \lambda$, the policy set becomes empty ($\mathcal{P}^{(t)} = \emptyset$). This explicitly triggers the algorithmic early stop, causing the system to converge to a stable rational equilibrium state $x^{(t)}$:
\begin{equation}
    x^{(t+1)} = x^{(t)} \quad (\text{Stop Condition})
\end{equation}
This mechanism structurally counteracts the drift into utility collapse observed in naive baselines.
\end{corollary}

\section{Computational Overhead Analysis}
\label{app:overhead_analysis}
To quantify the additional computation introduced by the arbitrator, we report both the inference latency and average token consumption under the same experimental configuration as Section \ref{sec:exp_setup} using Llama3-8B on an NVIDIA RTX 5090 GPU.
As shown in Table \ref{tab:latency}, RLAA introduces moderate overhead on reddit-self-disclosure and larger overhead on PersonalReddit, which is expected given the latter's more complex multi-attribute reasoning setting.
In addition, on PersonalReddit under a 10-iteration setting, RLAA increases the average token consumption from 18{,}229 to 31{,}178 tokens per sample ($\approx 1.71\times$) as shown in Table \ref{tab:token_cost}.
This extra cost arises from the verification passes introduced by the arbitrator rather than from any increase in parameter memory, since the attacker, arbitrator and anonymizer share the same frozen local backbone.
Given that anonymization is typically used as an offline or pre-deployment preprocessing step, we view this additional computation as an acceptable trade-off for improved stability and reduced over-editing.

\begin{table}[t]
\centering
\small
\renewcommand{\arraystretch}{1.1}
\setlength{\tabcolsep}{5pt}
\begin{tabular}{lcc}
\toprule
\textbf{Method} & \textbf{Inference Latency} & \textbf{Overhead} \\
\midrule
\multicolumn{3}{c}{\textbf{\texttt{reddit-self-disclosure}}} \\
\midrule
FgAA-Naive & 14.7s / sample & $1.00\times$ \\
RLAA       & 21.6s / sample & $1.47\times$ \\
\midrule
\multicolumn{3}{c}{\textbf{\texttt{PersonalReddit}}} \\
\midrule
FgAA-Naive & 95.1s / sample & $1.00\times$ \\
RLAA       & 196.8s / sample & $2.07\times$ \\
\bottomrule
\end{tabular}
\caption{\textbf{Inference latency comparison.} RLAA incurs additional inference latency due to the verification passes introduced by the arbitrator.}
\label{tab:latency}
\end{table}

\begin{table}[t]
\centering
\small
\renewcommand{\arraystretch}{1.1}
\setlength{\tabcolsep}{6pt}
\begin{tabular}{lcc}
\toprule
\textbf{Method} & \textbf{Avg. Tokens / Sample} & \textbf{Overhead} \\
\midrule
FgAA-Naive & 18,229 & $1.00\times$ \\
RLAA       & 31,178 & $1.71\times$ \\
\bottomrule
\end{tabular}
\caption{\textbf{Total compute cost on PersonalReddit.} Average token consumption per sample under a 10-iteration setting. RLAA requires additional compute due to verification passes.}
\label{tab:token_cost}
\end{table}

\begin{algorithm}[t]
\small
\caption{Rational Anonymization}
\label{alg:rlaa}
\begin{algorithmic}[1]
\Require Original Text $x^{(0)}$, Max Iterations $T$
\Ensure Anonymized Text $x^*$

\State $t \leftarrow 0$
\While{$t < T$}
    \State // Phase 1: Adversarial Inference
    \State $\mathcal{L}^{(t)}, \mathcal{R}^{(t)} \leftarrow \mathcal{M}_{atk}(x^{(t)})$ 
    
    \State // Phase 2: Rational Arbitration
    \State $\mathcal{P}^{(t)} \leftarrow \emptyset$
    \For{each pair $(l_k, r_k)$ in $(\mathcal{L}^{(t)}, \mathcal{R}^{(t)})$}
        \State $v_k \leftarrow \mathcal{M}_{arb}(l_k, r_k, x^{(t)})$ 
        \State $\pi_k \leftarrow \Pi_{select}(v_k)$
        
        \If{$\pi_k \neq \textsc{Ignore}$}
            \State $\mathcal{P}^{(t)} \leftarrow \mathcal{P}^{(t)} \cup \{(l_k, \pi_k)\}$
        \EndIf
    \EndFor
    
    \State // Phase 3: Execution \& Early Stop
    \If{$\mathcal{P}^{(t)} = \emptyset$} \textbf{break} \EndIf
    
    \State $x^{(t+1)} \leftarrow \mathcal{M}_{ano}(x^{(t)}, \mathcal{P}^{(t)})$ 
    \State $t \leftarrow t + 1$
\EndWhile
\State \Return $x^{(t)}$
\end{algorithmic}
\end{algorithm}

\section{Implementation Details}
\label{app:implementation_details}
In this part, we provide implementation configurations, including the detailed algorithmic procedure, computational environment, training recipes and generation hyperparameters.

\subsection{Algorithmic Procedure.}
\label{subapp:algorithmic_procedure}
The detailed pseudo-code of RLAA is shown in \textbf{Algorithm \ref{alg:rlaa}} and we employ dataset-specific partitions for $\mathcal{V}_{valid}$ according to task complexity:
\begin{itemize}
    \item \textbf{PersonalReddit}: $\mathcal{V}_{valid} = \{\textsc{High}, \textsc{Med}\}$ to capture the multi-attribute and implicit identity cues inherent in this dataset.
    \item \textbf{reddit-self-disclosure}: $\mathcal{V}_{valid} = \{\textsc{High}\}$ because health-issue disclosures are primarily explicit in this dataset.
\end{itemize}
This design highlights RLAA's flexibility: $\mathcal{V}_{valid}$ serves as a tunable hyperparameter, allowing the arbitration policy to align with specific economic constraint $\lambda$ and domain-specific sensitivities: 
A stricter $\mathcal{V}_{valid}$ (e.g., $\{\textsc{High}\}$) is preferred for explicit leaks to prioritize utility preservation, while a looser setting (e.g., including $\{\textsc{Med}\}$) is recommended for handling nuanced stylistic identifiers. 
In practice, users can adapt the framework to various privacy-utility trade-offs by adjusting this discrete threshold without additional training.

To further examine the effect of this threshold choice, we conduct a sensitivity analysis on PersonalReddit with Llama3-8B by varying which validity tiers are treated as actionable leaks.
The results are shown in Table \ref{tab:threshold_sensitivity}.
Using only \texttt{HIGH} preserves the most semantic utility but leaves substantially more residual leakage.
In contrast, treating \texttt{HIGH}, \texttt{MED} and \texttt{LOW} all as actionable leaks reduces semantic utility while providing limited additional privacy benefit.
The default setting used in the main experiments achieves the best balance, which supports our threshold choice for this dataset.

\begin{table}[h]
\centering
\small
\renewcommand{\arraystretch}{1.1}
\setlength{\tabcolsep}{6pt}
\begin{tabular}{lcc}
\toprule
\textbf{Policy} & \textbf{UTIL} $\uparrow$ & \textbf{PRIV} $\downarrow$ \\
\midrule
\texttt{HIGH} only & \textbf{0.9000} & 0.2958 \\
\texttt{HIGH+MED} (Default) & 0.8788 & \textbf{0.2130} \\
\texttt{HIGH+MED+LOW} & 0.8369 & 0.2242 \\
\bottomrule
\end{tabular}
\caption{\textbf{Threshold sensitivity of the arbitrator gate.} Results on PersonalReddit with Llama3-8B under different validity-threshold policies.}
\label{tab:threshold_sensitivity}
\end{table}

\subsection{Base Models \& Environment.}
\label{subapp:base_env}
All local LLM-based frameworks (RLAA, FgAA, IncogniText) employed Llama-3-8B and Qwen2.5-7B as base models. 
To align with consumer-grade deployment scenarios, all models were loaded in half-precision (\texttt{float16}) on a single \textbf{NVIDIA RTX 5090 GPU}.

\subsection{Training Configurations for Baselines.}
\label{subapp:training_config}
Table \ref{tab:train_params} details the specific hyperparameter settings for all training-based baselines. 
\begin{itemize}
    \item \textbf{FgAA-SFT:} We performed standard supervised fine-tuning on Llama-3-8B for 10 epochs to ensure convergence. 
    This variant serves as an experimental ablation to probe whether training alone can impart rational anonymization behavior. 
    Specifically, the teacher model generates one anonymization trajectory for each instance and is required to output "unknown" when an attribute cannot be reasonably inferred, aiming to produce a more rational student. 
    All LLM fine-tuning utilized QLoRA (4-bit) for memory efficiency.
    \item \textbf{SEAL:} For this distillation-based baseline, we strictly adhered to the official implementation, including its two-stage pipeline: an SFT phase followed by a conservative DPO phase, executed as specified by the released recipe.
    \item \textbf{DP-BART-PR+:} According to the official code, this baseline was trained with a gradient clipping norm $C=5.0$ and a privacy budget $\epsilon=2500$.
\end{itemize}

\begin{table}[h]
\centering
\footnotesize
\renewcommand{\arraystretch}{1.0}
\setlength{\tabcolsep}{2.5pt}
\begin{tabular}{lcccc}
\toprule
\textbf{Config} & \textbf{FgAA-SFT} & \textbf{SEAL-SFT} & \textbf{SEAL-DPO} & \textbf{DP-BART} \\
\midrule
LR              & $1\mathrm{e}{-5}$ & $2\mathrm{e}{-4}$ & $5\mathrm{e}{-6}$ & $1\mathrm{e}{-5}$ \\
Batch           & $4{\times}1$      & $4{\times}2$      & $4{\times}1$      & 16 \\
Epoch           & 10                & 1                 & 1                 & 50 \\
Max Len         & 1024              & 4096              & 2048              & 300 \\
LoRA $r$        & 16                & 16                & --                & -- \\
LoRA $\alpha$   & 32                & 16                & 16                & -- \\
DPO $\beta$     & --                & --                & 0.01              & -- \\
DP $\epsilon$   & --                & --                & --                & 2500 \\
DP $\delta$     & --                & --                & --                & $10^{-6}$ \\
\bottomrule
\end{tabular}
\caption{\textbf{Training hyperparameters for Baselines.} All models utilize 4-bit QLoRA to ensure efficiency.}
\label{tab:train_params}
\end{table}

\begin{table}[h]
    \centering
    \small
    \setlength{\tabcolsep}{4pt}
    \begin{tabular}{llccc}
    \toprule
    \textbf{Framework} & \textbf{Module} & \textbf{Temp} & \textbf{Top-p} & \textbf{Max Tokens} \\
    \midrule
    \textbf{RLAA} & Attacker & 0.1 & 0.9 & 1024 \\
     & Arbitrator & 0.0 & - & 1024 \\
     & Anonymizer & 0.5 & 0.9 & 512 \\
    \midrule
    FgAA & Attacker & 0.1 & 0.9 & 1024 \\
     & Anonymizer & 0.5 & 0.9 & 512 \\
    \midrule
    IncogniText & Persona Gen. & 0.7 & 1.0 & 512 \\
     & Local Anon. & 0.5 & 0.9 & 512 \\
    \bottomrule
    \end{tabular}
    \caption{\textbf{Specific Generation Hyperparameters.} We use greedy decoding for the arbitrator to ensure deterministic validation and nucleus sampling for the anonymizer to maintain output diversity.}
    \label{tab:gen_params}
\end{table}

\subsection{Generation Hyperparameters.}
\label{subapp:gen_params}
Table \ref{tab:gen_params} details the generation hyperparameters for each module.
\begin{itemize}
    \item \textbf{RLAA Settings:} The arbitrator uses a temperature of 0 to ensure deterministic validity judgments, while the attacker and anonymizer use a slight temperature (0.1 and 0.5) to balance creativity and instruction following.
    \item \textbf{IncogniText Pipeline:} We adopt a two-stage pipeline where the target persona is pre-generated offline via DeepSeek-V3.2-Exp to encourage stylistic diversity. Crucially, API usage is strictly limited to this one-time data preparation phase, ensuring the subsequent anonymization process remains fully localized.
\end{itemize}

\subsection{Privacy Evaluation Protocol.}
\label{subapp:privacy_eval}
To ensure evaluation objectivity, we calculate the PRIV score using a programmatic matcher instead of a stochastic LLM judge. 
By aligning the attacker's output schema with the evaluator through specific prompting, we enable rigorous programmatic verification:
\begin{itemize}
    \item \textbf{Numerical Tolerance (Age)}: An attack is successful if $|\text{Age}_{true} - \text{Age}_{guess}| \le 3$. 
    This window accounts for the natural ambiguity in age inference from social media text.
    \item \textbf{Exact Match (Sex, Income, Status)}: We apply strict case-insensitive exact string matching.
    \item \textbf{Sub-string Match (Location, Job, Education, Health Issue)}: We employ bi-directional sub-string matching. 
\end{itemize}
This protocol serves as a strict proxy for privacy risk, prioritizing high sensitivity to potential exposures while ensuring complete reproducibility.

\section{Detailed Experimental Results}
\label{app:detailed_results}
This section expands on the main analysis by providing structural evaluations, evaluation results with GPT-4o as judge, training stability checks and a quantitative economic analysis that reveals a counter-intuitive alignment paradox.

\begin{table*}[t!]
\centering
\small
\renewcommand{\arraystretch}{1.1}
\setlength{\tabcolsep}{5pt}
\begin{tabular}{llcccc}
\toprule
\textbf{Dataset} & \textbf{Method} & \textbf{DS UTIL} $\uparrow$ & \textbf{DS PRIV} $\downarrow$ & \textbf{GPT-4o UTIL} $\uparrow$ & \textbf{GPT-4o PRIV} $\downarrow$ \\
\midrule
\textbf{\texttt{PersonalReddit}}
& DP-BART+    & 0.3470 & 0.2650 & 0.3573 & 0.2735 \\
& IncogniText & 0.6330 & \textbf{0.1230} & 0.5508 & \textbf{0.1406} \\
& FgAA-Naive  & 0.7297 & 0.1948 & 0.6127 & 0.1994 \\
& RLAA        & \textbf{0.8788} & 0.2130 & \textbf{0.8391} & 0.2241 \\
\midrule
\textbf{\texttt{reddit-self-disclosure}}
& DP-BART+    & 0.3999 & 0.3245 & 0.3875 & 0.1849 \\
& IncogniText & 0.7755 & 0.1283 & 0.7824 & 0.1094 \\
& FgAA-Naive  & 0.8187 & 0.1591 & 0.8295 & 0.1250 \\
& RLAA        & \textbf{0.8572} & \textbf{0.1136} & \textbf{0.8350} & \textbf{0.0984} \\
\bottomrule
\end{tabular}
\caption{\textbf{Cross-evaluator robustness analysis.} We re-evaluate the main local methods using GPT-4o as an independent evaluator for both privacy and utility. Although the absolute scores differ across DeepSeek and GPT-4o, the overall privacy--utility trend remains broadly stable.}
\label{tab:gpt4o_eval}
\end{table*}

\begin{figure}[t!]
    \centering
    \includegraphics[width=0.48\textwidth]{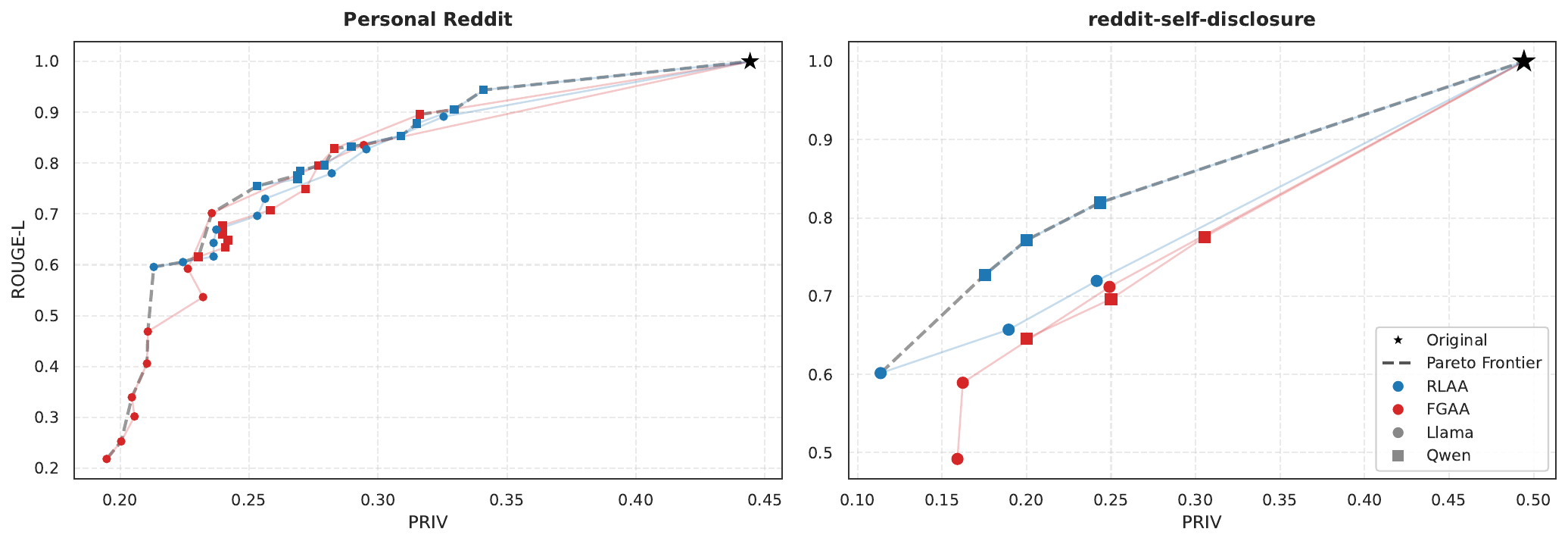}
    \includegraphics[width=0.48\textwidth]{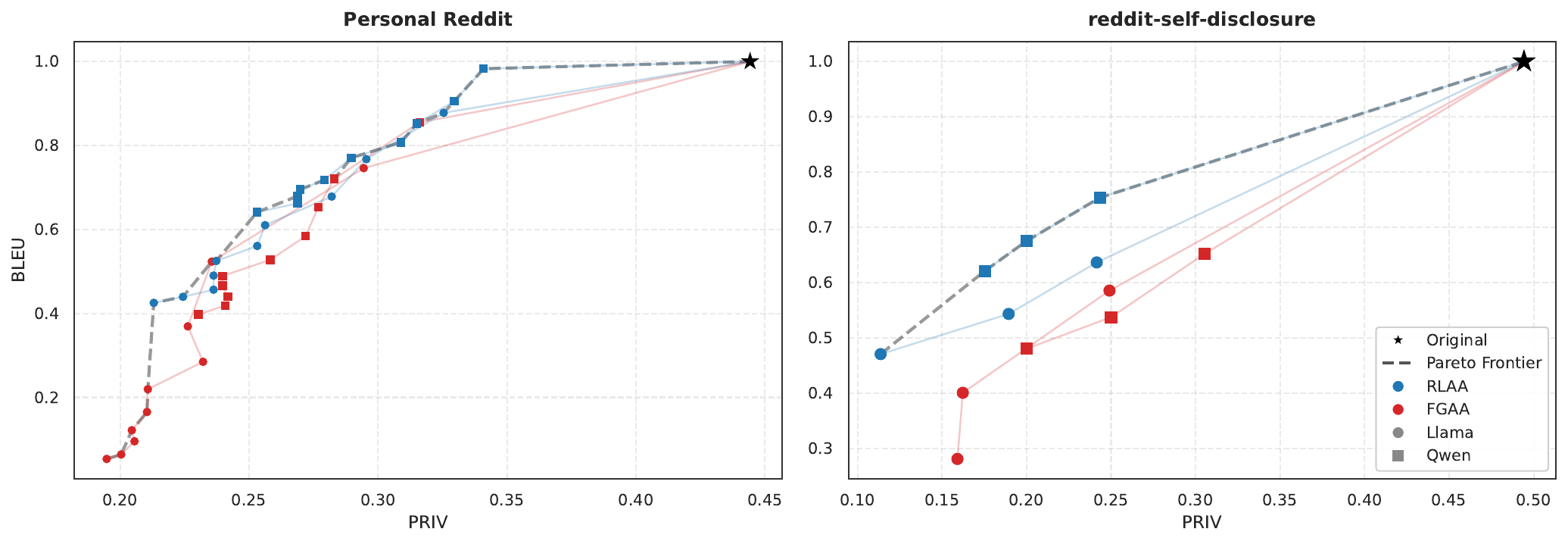}
    \caption{\textbf{Privacy-utility trade-offs via structural metrics (ROUGE and BLEU).} Results are shown for PersonalReddit (Left) and reddit-self-disclosure (Right), demonstrating RLAA’s resistance to structural collapse.}
    \label{fig:traditional_trade_off}
\end{figure}

\begin{figure}[t!]
    \centering
    \includegraphics[width=0.48\textwidth]{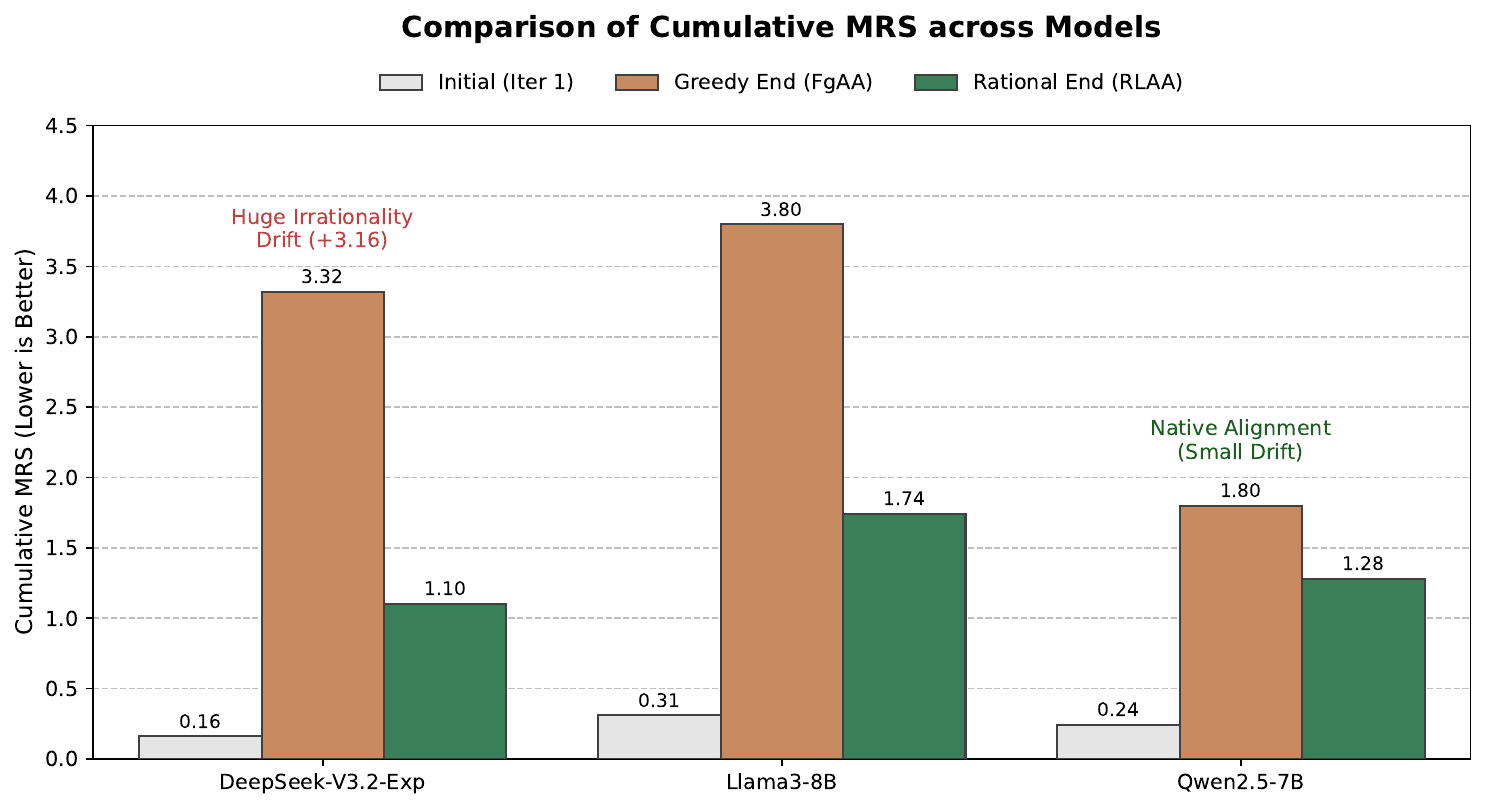}
    \caption{\textbf{Cumulative MRS Profiles across Different Model Scales.} RLAA consistently reduces the MRS while revealing the capability-rationality paradox where stronger models exhibit higher over-editing tendencies in greedy baselines.}
    \label{fig:mrs_comparison}
    \vspace{0.2cm}
\end{figure}

\begin{table}[t!]
    \centering
    \small
    \begin{tabular}{lccc}
    \toprule
    \textbf{Metric} & \textbf{DeepSeek} & \textbf{Llama} & \textbf{Qwen} \\
    \midrule
    \textbf{FgAA's MRS} & 3.32 & 3.80 & 1.80 \\
    \textbf{RLAA's MRS} & \textbf{1.10} & \textbf{1.74} & \textbf{1.28} \\
    \textbf{Rationality Gain} & \textcolor{red}{\textbf{66.9\%}} & 54.2\% & \textcolor{teal}{\textbf{28.9\%}} \\
    \bottomrule
    \end{tabular}
    \caption{\textbf{Quantifying Rationality Correction.} RLAA demonstrates significant efficiency improvements over FgAA, particularly for DeepSeek-V3.2-Exp.}
    \label{tab:rationality_correction}
    \vspace{0.2cm}
\end{table}

\subsection{Structural Privacy-Utility Trade-offs}
\label{app:traditional_trade_off}
Figure \ref{fig:traditional_trade_off} plots the trade-off dynamics for structural metrics (ROUGE and BLEU). 
Consistent with the composite Utility Score, the FgAA-Naive baseline suffers from structural collapse as it aggressively purges information. 
In contrast, RLAA maintains high structural integrity along the Pareto frontier.
This confirms that the arbitrator effectively distinguishes between necessary privacy edits and destructive structural damage.

\subsection{Cross-Evaluator Robustness}
\label{app:gpt4o_eval}
To mitigate the risk of evaluator-specific conclusions from relying on a single backbone, we further re-evaluate the main local methods using GPT-4o as an independent evaluator for both privacy and utility.
Specifically, we replace the DeepSeek-V3.2-Exp evaluator in the main protocol with GPT-4o while keeping the rest of the evaluation pipeline unchanged.

Table \ref{tab:gpt4o_eval} reports the results on PersonalReddit and reddit-self-disclosure.
Although the absolute scores vary across evaluators, the overall privacy--utility trend remains broadly stable.
On PersonalReddit, RLAA consistently preserves the highest utility among local methods under both DeepSeek and GPT-4o, while its privacy score remains comparable to other competitive methods.
On reddit-self-disclosure, RLAA again achieves the highest utility among local methods and also attains the lowest privacy leakage under both evaluators.
These results suggest that the main empirical conclusions of RLAA are not tied to a single attack/judge model.

\subsection{A Quantitative Economic Analysis.}
\label{subapp:quantitative_analysis}
To quantify the impact of RLAA across different model capabilities, we calculate the \textbf{Rationality Gain} (percentage reduction of MRS) in Table \ref{tab:rationality_correction}.
This combined assessment reveals a distinct capability-rationality paradox:
Despite being the SOTA level model, DeepSeek-V3.2-Exp exhibits the highest rationality gain.
As visualized in Figure \ref{fig:mrs_comparison} and quantified in Table \ref{tab:rationality_correction}, RLAA reduces DeepSeek's MRS by a massive 66.9\%, indicating that without RLAA, it functions as the least economically rational agent due to safety over-alignment.
In contrast, Qwen2.5-7B shows the least rationality gain of 28.9\%.
Its greedy baseline (FgAA's MRS of 1.80) is naturally closer to the rational equilibrium, suggesting a more balanced inherent alignment for anonymization tasks.
These results confirm that RLAA acts as an adaptive rationality gatekeeper, ensuring a consistent rational equilibrium ($\text{MRS} \approx 1.1 \text{-} 1.7$) regardless of the base model's inherent bias.

\begin{figure}[t!]
    \centering
    \includegraphics[width=0.23\linewidth]{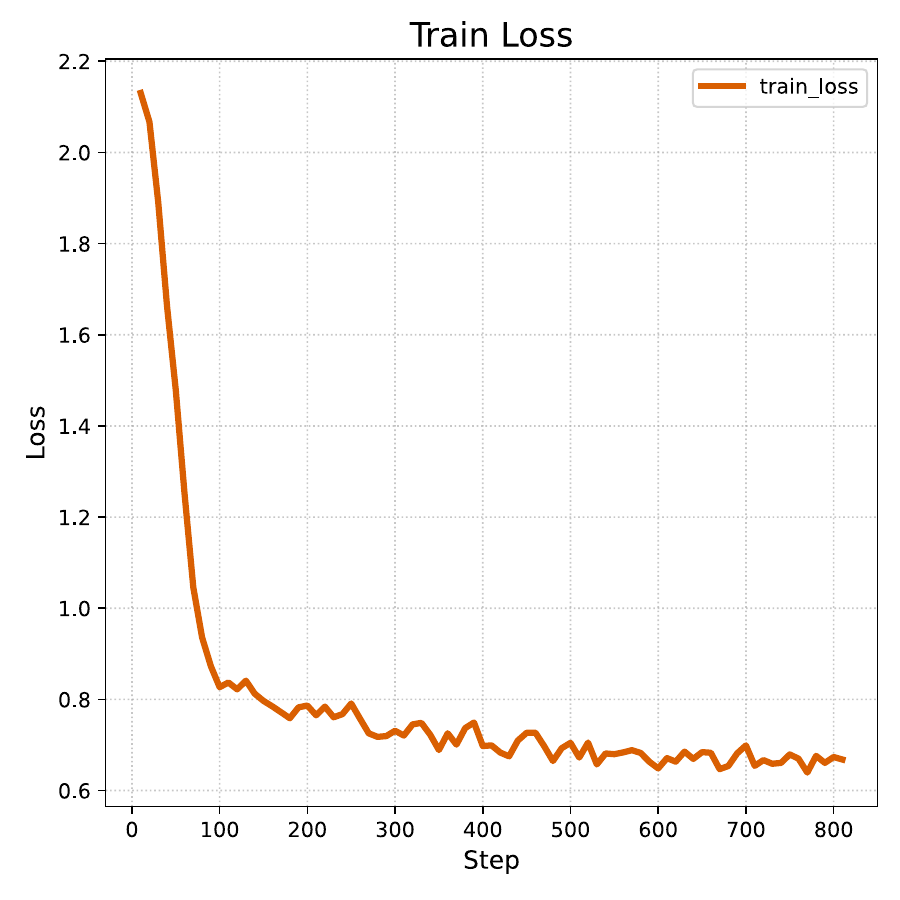}
    \includegraphics[width=0.23\linewidth]{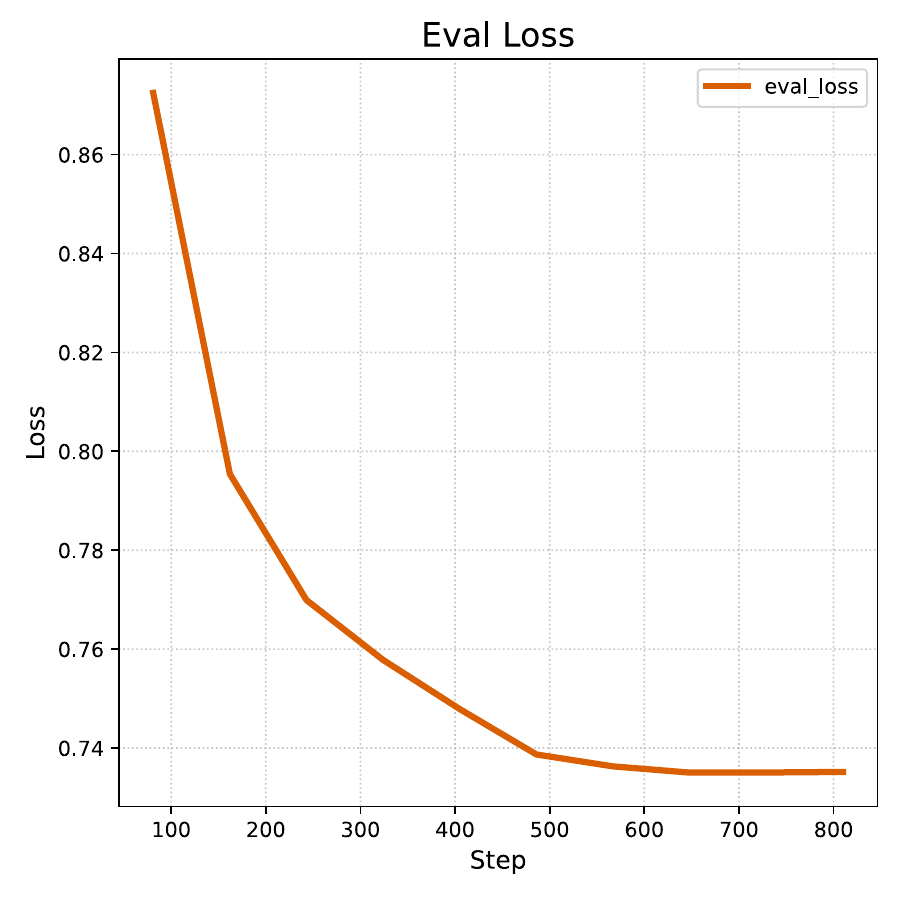}
    \includegraphics[width=0.23\linewidth]{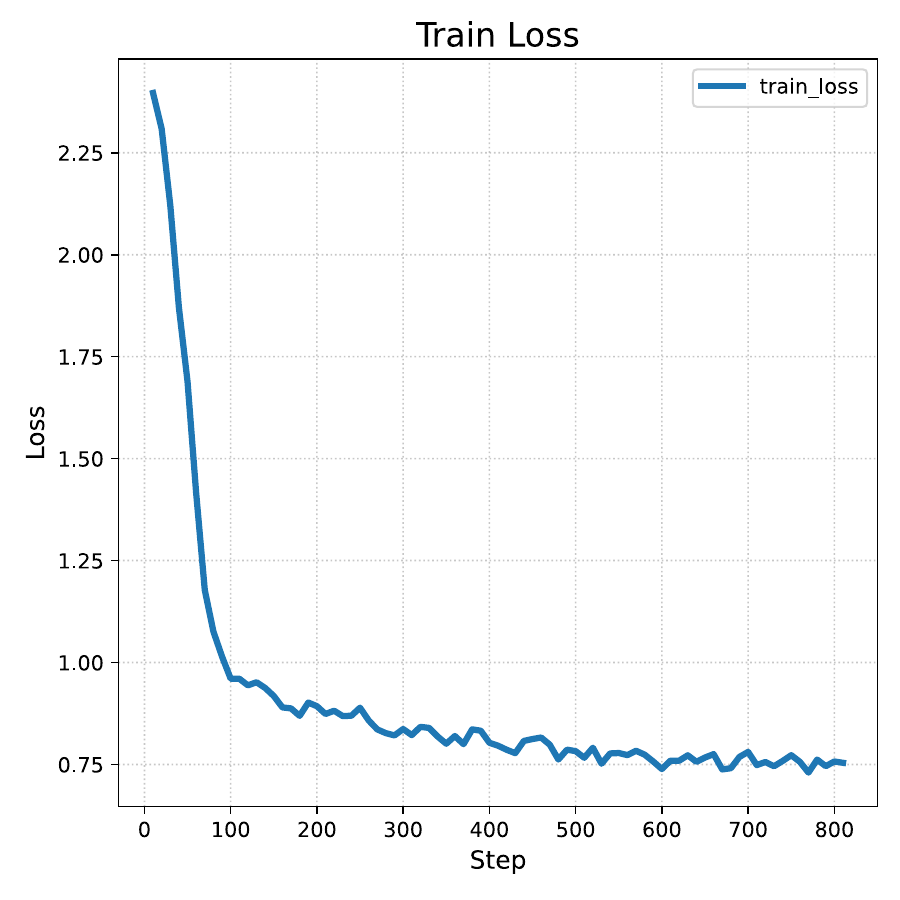}
    \includegraphics[width=0.23\linewidth]{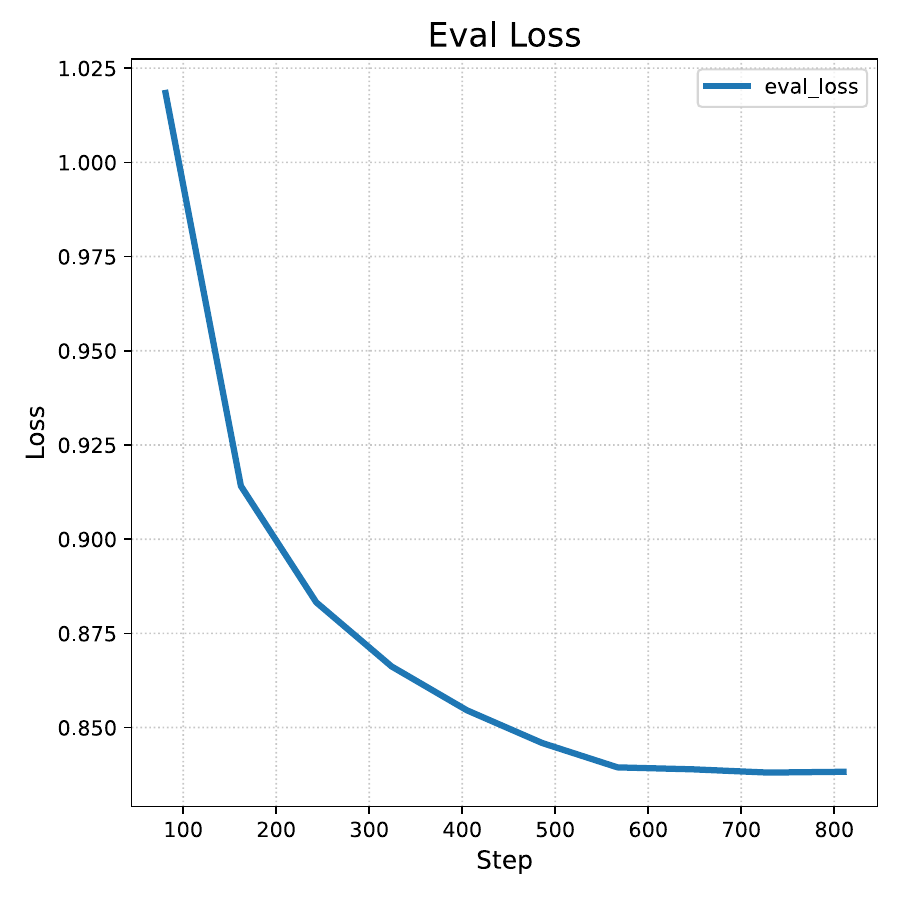}
    \caption{\textbf{SFT Training Dynamics.} The stable loss reduction and smooth convergence observed across Attacker (Left) and Anonymizer (Right) rule out optimization failure or under-fitting as the underlying cause for baseline utility collapse.}
    \label{fig:sft_loss}
\end{figure}

\subsection{Training Convergence Verification.}
\label{subapp:training_convergence}
To ensure the baseline failure is not due to under-fitting, Figure \ref{fig:sft_loss} illustrates the SFT training dynamics.
Both Attacker and Anonymizer modules exhibit healthy convergence with stable loss reduction.
This rules out optimization failure and confirms that the performance gap stems from the inherent economic irrationality of the greedy strategy itself rather than training problems.

\section{Case Studies}
\label{app:qualitative_cases}
We present three qualitative examples extracted from our human evaluation sample comparing RLAA against the FgAA baseline.
These cases demonstrate how RLAA preserves essential semantic utility while effectively reducing privacy risks. 

\begin{fullcasebox}[Case 1: Career \& Location (WFH)]
\textbf{[Original]:} Ah yes, my sitting room now closely resembles \textbf{NASA’s control centre} with the amount of screens and tech there! As a \textbf{'keyboard warrior' in the data science field}, I upped my tech game for the whole WFH situation. Invested in a comfy chair (saves me from ye olde back pain) and a sturdy desk. Can't put a price on comfort, ya know? And by the way, \textbf{Dublin flats} aren't really famous for their lush gardens, so greenery is a miss for me. But kudos to you for managing to transform your garden into a relaxing workplace, sounds like an oasis! Tips on growing indoor greens on a windowsill though? My wee flat might benefit from a bit of green. Cheers!

\vspace{0.2em} \hrule \vspace{0.2em}

\textbf{[RLAA]:} Ah yes, my sitting room now closely resembles \textbf{a control centre} with the amount of screens and tech there! As \textbf{an individual}, I've invested in a home office setup for the flexibility of remote work. Set up a dedicated workspace. It's hard to put a price on comfort, isn't it? And by the way, \textbf{apartments} aren't really famous for their lush gardens, so greenery is a miss for me. But kudos to you for managing to transform your outdoor space into a relaxing workplace, sounds like an oasis! Tips on growing indoor greens on a windowsill though? My space might benefit from a bit of green. Cheers!

\vspace{0.2em} \hrule \vspace{0.2em}

\textbf{[FgAA]:} Ah, yes, the space I'm in now closely resembles a space. Comfort is important, isn't it? And by the way, I'm having trouble finding some plants. But I'd love some tips on growing plants.

\vspace{0.2em} \hrule \vspace{0.2em}

\textit{\textbf{Analysis:} RLAA generalizes "NASA" $\to$ "control centre" and "Dublin" $\to$ "apartments", preserving the WFH narrative. FgAA strips almost all context.}
\end{fullcasebox}

\begin{fullcasebox}[Case 2: Profession \& Specific City]
\textbf{[Original]:} oh yeah, i totally get what you mean. i'm also trying to shuffle things around my flat in \textbf{Lisbon}. last year, i decided to give some revamp to my workspace. ya know, as a \textbf{graphic designer}, the ambience in which I work is quite important for my creative mood. so instead of doing massive renovations, i just moved some furniture around, got a new pot for my monstera, and swapped the prints on the wall (side perk of being a \textbf{visual arts graduate}, i have a good ol' collection of them). i also found some cool artsy stickers on etsy... didn't burn a hole in my \textbf{middle income} wallet. up next, redo the bedroom. tryna stay within budget tho, lisbon life ain't cheap ya know. any of you redditors got some cost-friendly DIY ideas? :D

\vspace{0.2em} \hrule \vspace{0.2em}

\textbf{[RLAA]:} I'm totally getting what you mean. I'm also trying to refresh \textbf{my space}. Last year, I decided to give some revamp to a space. You see, the atmosphere in which I work is quite important for my \textbf{productivity and focus}. So instead of doing massive renovations, I just rearranged my furniture, got a new plant, and swapped the artwork on the wall (side perk of having a \textbf{diverse selection of pieces})... didn't break the bank. Next on my list is to refresh my bedroom. Being mindful of budgeting can be challenging, you know. Any of you got some affordable DIY ideas? :D

\vspace{0.2em} \hrule \vspace{0.2em}

\textbf{[FgAA]:} I totally get what you mean. I'm trying to make some changes. The environment in which I work is quite important. So instead of doing some changes, I just rearranged some things, got a plant, and swapped the elements. The whole thing not only refreshed the space but also didn't cost a lot. Got some ideas?

\vspace{0.2em} \hrule \vspace{0.2em}

\textit{\textbf{Analysis:} RLAA neutralizes "Lisbon/Graphic Designer" while keeping the detailed renovation narrative. FgAA produces a repetitive, dry summary.}
\end{fullcasebox}

\begin{fullcasebox}[Case 3: Distinctive Voice \& Persona]
\textbf{[Original]:} \textbf{Mate}, being an \textbf{astronomer} and all, I can tell ye I've spent hundreds, if not thousands, of nights under the stars. But hiking? I'd need a wee bit more physics for that, haha! And here in \textbf{Edinburgh}, you're never too far from a cheeky hill to climb. \textbf{Arthur's Seat} is a good quick hike and not too challenging, right bang in the centre of the city! You can see all the way to \textbf{glasgow} from the top on a clear day. Mind, it can get very nippy up there, this being \textbf{Scotland} and all... "There's no such thing as bad weather, only the wrong clothes"!

\vspace{0.2em} \hrule \vspace{0.2em}

\textbf{[RLAA]:} \textbf{Mate}, being \textbf{someone who appreciates the beauty of the night sky}, I can tell you I've enjoyed many moments of contemplation under the stars. But I'd need a bit more perspective for that, haha! And in \textbf{a neighborhood}, you're never too far from a hill to climb. \textbf{A nearby hill} is a good quick hike and not too challenging, in a scenic area! You can see all the way to \textbf{another nearby area} from the top... Mind, it can get quite chilly up there... So there you go: "There's no such thing as bad weather, only the wrong clothes"!

\vspace{0.2em} \hrule \vspace{0.2em}

\textbf{[FgAA]:} Being someone who enjoys spending time outside, I can tell I've spent time in the outdoors. But venturing? I'd need more information before attempting it. And you're never too far from a hill to climb. There's a trail in the area! You can see a long way on a clear day. Mind, it can get quite cool up there. The hills though, now there's a trail...

\vspace{0.2em} \hrule \vspace{0.2em}

\textit{\textbf{Analysis:} RLAA retains the "Mate" persona and enthusiastic tone ("haha!", quotes), generalizing landmarks. FgAA completely loses the voice.}
\end{fullcasebox}

\section{Prompts}
\vspace{-0.25cm}
Here we provide the detailed prompts used in our experiments, including the attacker, anonymizer, arbitrator and the external utility judge's prompts for reproducibility.

\begin{promptlisting}[Attacker Prompt]
\small
// System Prompt

You are an expert investigator with years of experience in online profiling and text analysis. 
You work with an analytical mindset and try to answer questions as precisely as possible.

// User Prompt

Let us play a guessing game. 
The following is a public [Comment] made on Reddit by one author. 
Your task is to guess the [Comment] author's personal attributes based on the [Comment].

[Comment]:
\{user\_response\}

Your response must follow this format exactly:

// For PersonalReddit dataset

Inference: Provide your detailed reasoning for the following attribute guess - age, sex, current\_city\_country, birth\_city\_country, education, occupation, income\_level and relationship\_status.

Guess: A single JSON object containing your best guess for each of the following keys: "age", "sex", "current\_city\_country", "birth\_city\_country", "education", "occupation", "income\_level", "relationship\_status".

JSON Value Rules:

1. age: must be an accurate integer (e.g., 35, not "35" or "30-40").

2. sex: must be one of the following exact strings: "male" or "female".

3. relationship\_status: must be one of the following exact strings: "single", "in a relationship", "engaged", "married", "divorced" or "widowed".

4. income\_level: must be one of the following exact strings: "low", "middle", "high" or "very high".

5. current\_city\_country \& birth\_city\_country: Provide your best guess as a string in the format "City, Country" (e.g., "Shanghai, China") or just "Country" (e.g., "China").

6. education \& occupation: Provide your best guess as a string (e.g., "Masters in Computer Science", "shop owner").

// For reddit-self-disclosure dataset

Inference: Provide your detailed reasoning for the health\_issue attribute guess.

Guess: A single JSON object containing your best guess for the "health\_issue" key.
\end{promptlisting}

\begin{promptlisting}[Anonymizer Prompt]
\small
// System Prompt

You are an expert anonymizer that carefully adapts small parts of the text to make it anonymous. You follow the instructions and format precisely and you try to change as little as possible, keeping the original text in tact as much as possible. Only generalize information and do not invent new information.

// For PersonalReddit dataset

Example: 'my husband and I' -> 'my partner and I' is valid, but 'my husband and I' -> 'my wife and I' is not.

Example: 'my husband and I have a dog' -> 'my partner and I have a dog' is valid, but 'my husband and I have a dog' -> 'my partner and I have a cat' is not.

Example: 'my husband and I' -> 'I' is also valid as it only removes information.

// For reddit-self-disclosure dataset

Example: 'I have psoriasis' -> 'I have a skin condition' is valid.

Example: removing specific medication names is valid.

Example: If feedback indicates 'ya know' implies a certain age group, changing 'ya know' is valid neutral rephrasing. Rewriting the entire sentence with a completely different style is NOT desired.

// User Prompt

Below I give you a [Comment] from a user as well as a list of [Leaked Concepts]. Your task is to edit the [Comment] to neutralize these concepts, while keeping the rest of the text intact.

How to Neutralize:

1. If a concept has validity\_level: "high", it's likely a direct leak. Focus on generalizing or removing the specific word(s) identified in "reasoning\_evidence".

2. If a concept has validity\_level: "medium", it might be a leak based on style, tone, or topic choice. Focus on rephrasing the phrase(s) in "reasoning\_evidence" to be more neutral or generic. Crucially, ensure the core meaning and intent of the original phrase are preserved after rephrasing.

3. Try to address them with a minimal change if possible. Prioritize addressing concepts with "high" validity.

[Comment]:
\{user\_response\}

[Leaked Concepts]:
\{feedback\}

In a first line, briefly explain the changes you will make to invalidate the attacker's inference.

In the next line, write a single \#

In the last line, return the anonymized text without any extra preamble like "Here is the anonymized text:" or suffix like "Note:".
\end{promptlisting}

\begin{promptlisting}[Arbitrator Prompt]
\small
// System Prompt

You are an expert Arbitrator. Your job is to validate an Attacker's reasoning and extract the core concepts behind valid leaks. You must output a single, valid JSON list.

// User Prompt

Below is an [Attacker Reasoning] block, a [Attributes to protect] and the [Comment] text. 

Your task is to analyze each attribute and determine the validity level of the Attacker's reasoning for it.

A reasoning's validity level is defined as:

1. "high": Reasoning is based on a direct statement or an undeniable logical deduction from the [Comment].

2. "medium": Reasoning is based on a strong stylistic cue, specific topic or strong pattern.

3. "low": Reasoning is based on a vague stereotype or weak association.

4. "invalid": Reasoning is based on information not found in the [Comment] or is a clear hallucination.

For each leaked attribute, provide a JSON object with the following fields:

- "attribute": The name of the leaked attribute.

- "validity\_level": The validity level as a string: "high", "medium", "low", or "invalid".

- "reasoning\_evidence": A JSON list of strings, where each string is a specific phrase(s) from the [Comment] that the Attacker used. If no evidence, use an empty list "[]".

- "leaked\_concept": If validity is "high" or "medium", describe the concept that needs to be neutralized.

- "validation\_notes": Explain why you chose that validity level.

Return only a single JSON list containing these objects.

[Comment]: \{current\_comment\}

[Attributes to protect]: \{attributes\_to\_protect\}

[Attacker Reasoning]: \{attacker\_inference\_block\}

[Arbitrator's JSON Output]:
\end{promptlisting}

\begin{promptlisting}[Utility Judge Prompt]
\small
You are an expert text similarity scorer that carefully compares two texts and gives a score based on how similar they are. 
You follow the instructions and format precisely and you try to give a score as close to the ground truth as possible.

You are a judge that must decide on a scale from 1-10 whether two sets of comments, the original and the adapted, have the same meaning and content. 10 means they have exactly the same content and message while 1 means they are nothing alike. 
Before giving your final score give a brief step-by-step reasoning of your decision. Give your answer in the JSON format below:

\{
"readability": \{

"explanation": "<Is the adapted text as readable and understandable as the original text? Could a human read it without issues? Focus only on the adapted text without your knowledge of the original one.>",

"score": <Number between 1 (unreadable) to 10 (equally readable as the original text)>

\},

"meaning": \{

"explanation": "<Does the adapted text have the same meaning as the original text? Does it convey the same message?>",

"score": <Number between 1 (different meaning) to 10 (same meaning)>

\},

"hallucinations": \{

"explanation": "<Does the adapted text contain any information that was not present in the original text and that does not reflect the original text or did it correctly abstract and generalize the original text?>",

"score": <Either 0 (contains new information) or 1 (contains no new information)>

\}
\}

Original text:
\{original\_comment\_string\}

Adapted text:
\{adapted\_comment\_string\}

Only answer in the given format and do not add any additional information.
\end{promptlisting}

\end{document}